%% file: main.tex
\newif\ifRAL
\RALfalse 	

\newif\ifTRO
\TROtrue 	

\ifRAL
    \documentclass[letterpaper, 10 pt, journal, twoside]{IEEEtran}
\else
    \ifTRO
        \documentclass[lettersize,journal]{IEEEtran}
    \else
        \documentclass[letterpaper, 10 pt, conference]{IEEEconf}
    \fi
\fi

\IEEEoverridecommandlockouts

\ifRAL
\else
    \ifTRO
    \else
    \fi
\fi

\usepackage{graphicx} 

\usepackage[english]{babel}

\usepackage{amsmath} 
\usepackage{amssymb}  
\usepackage{amsthm}
\usepackage{amsfonts}
\usepackage{mathtools}
\usepackage{verbatim}
\usepackage[normalem]{ulem} 

\usepackage{color}
\usepackage[usenames,dvipsnames,table]{xcolor}
\usepackage{todonotes}
\usepackage{cite}
\usepackage{caption}
\usepackage{subcaption}  
\usepackage{cancel} 

\usepackage{comment}

\usepackage{booktabs}       

\usepackage{bm}
\providecommand{\bm}{\pmb}

\usepackage{tikz}

\usepackage[acronym]{glossaries}

\input{symbol_definition}

\title{STL-GCS: A Planner-Controller Framework for Signal Temporal Logic via Graphs of Time-varying Convex Sets}
\author{Nicola De Carli, Gregorio Marchesini, Dimos V. Dimarogonas}

\begin{document}

\markboth{Journal of \LaTeX\ Class Files,~Vol.~18, No.~9, September~2020}%
{How to Use the IEEEtran \LaTeX \ Templates}
\maketitle

\begin{abstract}
\input{Sections/abstract}
\end{abstract}

\ifRAL
\else
    \ifTRO
    \else
    \fi
\fi


\section{Introduction}

\input{Sections/intro}

\input{Sections/sys_model}

\input{Sections/stl_intro}

\input{Sections/stl_tv_sets}

\input{Sections/stl_gcs}

\input{Sections/stl_control}

\input{Sections/exp_results}

\input{Sections/conclusions}

\appendix

\input{Sections/app_selectingkappa}

\input{Sections/app_overlapping}
\bibliographystyle{unsrt}
\bibliography{biblio}

\end{document}

%% file: symbol_definition.tex
\DeclareMathOperator{\convhull}{co}

\newcommand{\nR}[1]{\mathbb{R}^{#1}}		
\newcommand{\Rpluseq}{\mathbb{R}_{\geq0}}		
\newcommand{\nN}[1]{\mathbb{N}^{#1}}		

\newcommand{\vect}[1]{\ensuremath{\bm{#1}}}		
\newcommand{\matr}[1]{\ensuremath{\bm{#1}}}		

\newcommand{\define}{\coloneqq}			
\newcommand{\norm}[1]{\left\lVert#1\right\rVert}

\newcommand{\interval}[2]{ \left[ #1, #2 \right] }
\newcommand{\rangen}[1]{ \ensuremath{\left[ #1 \right]} }

\newcommand{\classK}{\mathcal{K}}

\newcommand{\union}{\cup}
\newcommand{\bigunion}{\bigcup}

\newcommand{\fig}{Fig.~}	

\newcommand{\zeros}[1]{\matr{0}_{#1}}
\newcommand{\ones}[1]{\matr{1}_{#1}}

\newcommand{\pos}{\vect{p}}				
\newcommand{\rotMat}{\matr{R}}				

\newcommand{\state}{\boldsymbol{x}}
\newcommand{\sysinput}{\boldsymbol{u}}
\newcommand{\dstate}{\dot{\boldsymbol{x}}}

\newcommand{\error}{\boldsymbol{e}}

\newcommand{\planState}{\boldsymbol{z}}

\newcommand{\feedlinT}{\boldsymbol{\Pi}}

\newcommand{\planInput}{\boldsymbol{v}}

\newcommand{\config}{\boldsymbol{q}}

\newcommand{\ddconfig}{\ddot{\boldsymbol{q}}}

\newcommand{\configSet}{\mathcal{Q}}

\newcommand{\pathvar}{\ensuremath{s}}

\newcommand{\timelaw}{\ensuremath{p}}
\newcommand{\spacelaw}{\ensuremath{\rr}}

\newcommand{\stateSet}{\mathcal{X}}
\newcommand{\inputSet}{\mathcal{U}}

\newcommand{\spatiotemp}{\boldsymbol{\xi}}

\newcommand{\graph}{\ensuremath{\mathcal{G}}}
\newcommand{\edges}{\ensuremath{\mathcal{E}}}
\newcommand{\vertices}{\ensuremath{\mathcal{V}}}
\newcommand{\vertexv}{\ensuremath{v}}

\newcommand{\true}{\top}
\newcommand{\false}{\bot}
\newcommand{\pred}[1]{\mu^{#1}}
\newcommand{\formulaNT}{\varphi}
\newcommand{\formulaT}{\phi}
\newcommand{\formulaOR}{\psi}
\newcommand{\tempop}{\mathsf{T}}

\newcommand{\always}{\mathsf{G}}
\newcommand{\eventually}{\mathsf{F}}
\newcommand{\untilab}[2]{\mathsf{U}_{[{#1},{#2}]}}
\newcommand{\alwaysab}[2]{\mathsf{G}_{[{#1},{#2}]}}
\newcommand{\eventuallyab}[2]{\mathsf{F}_{[{#1},{#2}]}}
\newcommand{\tempopab}[2]{\mathsf{T}_{[{#1},{#2}]}}

\newcommand{\stlsatisfies}{\models}

\newcommand{\stldef}{\mathrel{::=}}
\newcommand{\stlsep}{\;\mid\;}

\newcommand{\robustness}[1]{\ensuremath{\rho}^{#1}}

\newcommand{\predFun}{\ensuremath{h}}

\newcommand{\tvset}{\ensuremath{\mathcal{B}}}
\newcommand{\tvsetspatiotemp}{\ensuremath{\bar{\mathcal{B}}}}
\newcommand{\predset}{\ensuremath{\mathcal{H}}}
\newcommand{\robpredset}{\ensuremath{\predset^r}}

\newcommand{\stlcbf}[1]{\ensuremath{\mathfrak{b}^{#1}}}

\newcommand{\robustMargin}{r}

\newcommand{\alwaystxt}{\emph{always}}
\newcommand{\eventuallytxt}{\emph{eventually}}
\newcommand{\untiltxt}{\emph{until}}

\newcommand{\cc}{\vect{c}}
\newcommand{\dd}{\vect{d}}
\newcommand{\ff}{\vect{f}}
\newcommand{\pp}{\vect{p}}
\newcommand{\rr}{\vect{r}}

\newcommand{\cC}{\vect{C}}
\newcommand{\gG}{\vect{G}}

\newcommand{\bgamma}{\vect{\gamma}}

\newtheorem{theorem}{Theorem}
\newtheorem{remark}[theorem]{Remark}
\newtheorem{proposition}[theorem]{Proposition}

\newtheorem{definition}[theorem]{Definition}
\newtheorem{example}{Example}

\newcommand{\xdim}{\ensuremath{n_x}}
\newcommand{\udim}{\ensuremath{n_u}}
\newcommand{\flatoutdegree}{\ensuremath{p}}

\newcommand{\planmodeldim}{\ensuremath{n_z}}
\newcommand{\configdim}{\ensuremath{d}}
\newcommand{\reldeg}{\varrho}

\newcommand{\numconsh}{\ensuremath{n_{h}}}   

\newcommand{\branchidx}{i}
\newcommand{\branchnum}{I}

\newcommand{\obsReg}{\mathcal{O}}
\newcommand{\freespace}{\mathcal{F}}
\newcommand{\colfreeReg}[1]{\mathcal{F}_{#1}}
\newcommand{\spcolfreeReg}[1]{\bar{\mathcal{F}}_{#1}}
\newcommand{\colfreeNum}{F}
\newcommand{\colfreeID}{\ensuremath{f}}

\newcommand{\freeregID}{\ell}
\newcommand{\freeregnum}{n_f}

\newcommand{\obsfun}[1]{o_{#1}}
\newcommand{\obsnum}{n_o}
\newcommand{\obsID}{l}
\newcommand{\obscons}[1]{n_{c,#1}}
\newcommand{\obsconsID}{i}

\newcommand{\tinit}{t_0}
\newcommand{\tfinal}{t_f}
\newcommand{\tinittask}[1]{t_{0#1}}
\newcommand{\tfinaltask}[1]{t_{f#1}}

\newcommand{\totallenght}{L}

\newcommand{\taskID}{\ensuremath{k}}
\newcommand{\tasknum}{\ensuremath{K}}

\newcommand{\funnelset}[1]{\tvsetspatiotemp^F_{#1}}
\newcommand{\predicateset}[1]{\tvsetspatiotemp^P_{#1}}
\newcommand{\funnelsetslice}[1]{\tvset^F_{#1}}

\newcommand{\lintermtv}{\ensuremath{\kappa}}

\newcommand{\bspline}{B-spline}
\newcommand{\numctrlpoints}{\ensuremath{J}}
\newcommand{\ctrlpointidx}{\ensuremath{j}}
\newcommand{\bcurve}{\ensuremath{\bgamma}}

\newcommand{\sourceV}{\ensuremath{\vertexv_s}}
\newcommand{\targetV}{\ensuremath{\vertexv_t}}

\newcommand{\path}{\mathfrak{p}}
\newcommand{\pathsSet}{\mathcal{P}}
\newcommand{\vertexvar}[1]{\mathbf{x}_{#1}}
\newcommand{\vertexset}[1]{\mathcal{X}_{#1}}

\newcommand{\taskset}{\ensuremath{{\mathcal{B}}}}
\newcommand{\sttaskset}{\ensuremath{\bar{\mathcal{B}}}}
\newcommand{\gcsset}{\ensuremath{\bar{\mathcal{S}}}}

\newcommand{\redG}{\bar{\boldsymbol{G}}}
\newcommand{\yaw}{\psi}
\newcommand{\rotz}{\rotMat_z}
\newcommand{\inertiaz}{\ensuremath{I_z}}

\newcommand{\force}{\boldsymbol{f}_u}
\newcommand{\torquez}{\tau_z}
\newcommand{\thrustersforces}{\boldsymbol{w}}
\newcommand{\allocmat}{\boldsymbol{B}_u}

\newcommand{\ND}[1]{{\color{magenta}#1\textsuperscript{NDC}}}

\newcommand{\greg}[1]{\textcolor{blue}{#1}}
\newcommand{\gregterminator}[1]{\textcolor{red}{\sout{#1}}}

\newcommand{\todoing}[1]{\todo[inline,color=blue!20, linecolor=orange!250]{\small#1}}

\newcommand{\gregsay}[1]{\todoing{\textbf{Greg98:} #1}}

%% file: Sections/abstract.tex
We present a unified trajectory planning and control framework for the satisfaction of Signal Temporal Logic (STL) specifications defined over convex predicates. At the planning layer, STL tasks are encoded as time-varying convex sets in configuration space, specifically designed so that forward invariance of the system with respect to these sets implies satisfaction of the specification with a prescribed robustness margin. This representation is then lifted to the joint time--configuration space and combined with the Graphs of Convex Sets (GCS) framework, yielding a shortest-path formulation of the planning problem over convex spatio-temporal sets. Trajectories are parameterized by B-splines, which enable continuous-time enforcement of STL satisfaction, collision avoidance, and smoothness constraints. At the control layer, the same time-varying sets used for planning are exploited to design a feedback controller that tracks the planned trajectory while prioritizing satisfaction of the STL specification during execution in the presence of tracking errors and model mismatch. We validate the proposed approach in simulation and in real-world experiments on space robotic platforms.

%% file: Sections/intro.tex
Robotic and autonomous systems are increasingly expected to accomplish complex missions while operating in cluttered environments and under precise time and space constraints\cite{lindemann2025formal,plaku2015motion, kress2018synthesis, silano2021power}. Examples include reaching a region within a prescribed time window, repeatedly visiting a charging station for a given duration, always staying in a safe region, or executing one among several admissible tasks.  
A substantial body of work has addressed motion planning under temporal specifications expressed in Linear Temporal Logic (LTL), which is well suited to encoding rich sequencing and persistence requirements \cite{kress2009temporal, belta2017formal}. More recently, Signal Temporal Logic (STL) has emerged as a particularly attractive formalism for robotic task specification. Unlike other temporal logics, STL combines logical formulas with continuous-time temporal operators, thereby enabling the specification of spatio-temporal requirements over continuous-time systems, and is equipped with quantitative semantics that measure how robustly a trajectory satisfies a given specification \cite{maler2004monitoring, belta2019formal, lindemann2025formal}.
These features make STL especially appealing in robotics, where robustness margins provide a natural way to reason about modeling uncertainty, disturbances, and tracking errors. However, synthesizing dynamically feasible trajectories that satisfy STL specifications remains highly challenging \cite{belta2019formal}.

Trajectory synthesis under STL constraints is often addressed by transcribing the underlying spatio-temporal constraints into finite-dimensional optimization problems. A prominent class of approaches relies on mixed-integer encodings: for linear systems, this typically leads to Mixed-Integer Linear Programs (MILPs), which provide sound and complete encodings of the STL satisfaction problem, but at the cost of exponential worst-case complexity, thereby limiting scalability \cite{belta2019formal,kurtz2022mixed,sadraddini2018formal,raman2015reactive,haghighi2019control,verhagen2024temporally}. 

Alternative optimization-based approaches treat STL synthesis more directly as a trajectory-optimization or optimal-control problem. These include methods based on smooth robustness surrogates and gradient-based optimization \cite{belta2019formal}, convex-concave or successive-convexification schemes \cite{takayama2025stlccp,uzun2026successive}, and sampling-based optimal-control methods such as MPPI for STL costs or minimum-violation planning \cite{halder2026lexicographic}. Such methods can avoid an explicit mixed-integer encoding of all logical decisions and can be integrated into standard optimal-control pipelines. However, they generally lead to nonconvex optimization problems and rely on local or approximate solution procedures; in particular, gradient-based methods may suffer from slow convergence and sensitivity to initialization, especially when the STL robustness landscape is nonsmooth or poorly conditioned. Moreover, these approaches typically generate an open-loop trajectory, with no explicit mechanism to preserve STL satisfaction during closed-loop execution.


\begin{figure}[t!]
    \centering
    \includegraphics[width=0.8\linewidth]{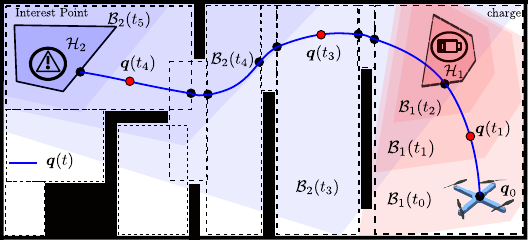}
    \caption{Trajectory $\config(t)$ of a quadrotor navigating through collision-free convex regions (dashed rectangles) while satisfying an STL specification by ensuring forward invariance of convex time-varying sets. The evolution of these sets is illustrated by the zero-level set boundaries of their defining functions, shown with different shades of red and blue. Black dots mark the junctions between consecutive B-splines, while red dots denote the robot configuration at the times corresponding to the displayed set of snapshots.}
    \label{fig:stl_example}
\end{figure}

Beyond MILP formulations, control theoretic methods have recently emerged as an effective alternative for enforcing STL specifications by directly incorporating STL constraints into analytical or computationally light feedback control laws\cite{lindemann2025formal,lindemann2017prescribed,lindemann2018control,buyukkocak2024sequential,buyukkocak2025resilient,rousseas2026operator}. In essence, such approaches rely on the ability to construct a time-varying set whose forward invariance implies satisfaction of the STL specification. Controllers achieving this objective are commonly designed within either the Control Barrier Function (CBF) framework \cite{ames2016control, lindemann2018control} or the Prescribed Performance Control (PPC) \cite{lindemann2017prescribed} framework. The benefits of such a method are twofold. On the one hand, it can enable direct feedback synthesis without an explicit planning stage, thus avoiding the computational burden of mixed-integer trajectory generation, although typically for restricted STL fragments and relatively simple environments. On the other hand, because the specification is enforced through feedback, these methods inherently provide greater robustness to disturbances and model mismatch. This is especially relevant when planning is performed using simplified system models to improve tractability.

When considering the general problem of trajectory planning without complex STL constraints, the framework of Graphs of Convex Sets (GCS), introduced by \cite{marcucci2024shortest,marcucci2023motion}, has recently shown remarkable performance in synthesizing smooth trajectories in nonconvex environments \cite{marcucci2023motion, pries2026admm}. The main idea is to decompose a nonconvex feasible region into a collection of convex subsets and encode their adjacency through a graph, so that planning can be cast as a shortest-path problem over a graph of convex sets. Remarkably, this representation preserves much of the combinatorial structure of the original trajectory planning problem while admitting a tight convex relaxation, thus enabling the use of efficient off-the-shelf convex optimization tools. Although the resulting mixed-integer problem is generally solved only approximately, bounds on sub-optimality of the relaxation are available \cite{marcucci2023motion}, and an integer-feasible solution is guaranteed to satisfy the original specification.


\subsubsection*{\textbf{Contributions}}
We propose a unified planning and control framework for systems subject to STL specifications by combining a time-varying-set reformulation of STL constraints with the Graphs of Convex Sets (GCS) framework.
We focus on a rich fragment of STL defined over convex predicates in configuration space. In the spirit of \cite{lindemann2018control}, we encode STL tasks as time-varying sets whose forward invariance guarantees satisfaction with a prescribed robustness margin. This representation plays a dual role. At the control level, it provides a sufficient condition for STL satisfaction under feedback. At the planning level, it restricts the feasible trajectories to evolve over a set of convex regions in the joint space and time domain. These regions are organized into a directed \emph{task graph}, whose edges encode the admissible temporal progression of the specification (see \fig\ref{fig:gsc_stl_example} for a simple example). Combining this graph with a \emph{collision-free graph} induced by a convex decomposition of the free space yields a \emph{product graph} that captures both task progression and obstacle-avoidance constraints. This discrete structure is then endowed with continuous decision variables and convex constraints, yielding a GCS formulation of the planning problem as a shortest-path problem over convex spatio-temporal sets. The resulting trajectories are parameterized by B-splines, enabling continuous-time enforcement of set containment, smoothness, and robustness margins (as summarized in \fig\ref{fig:stl_example}).

\begin{figure}[t!]
    \centering
    \includegraphics[width=0.8\linewidth]{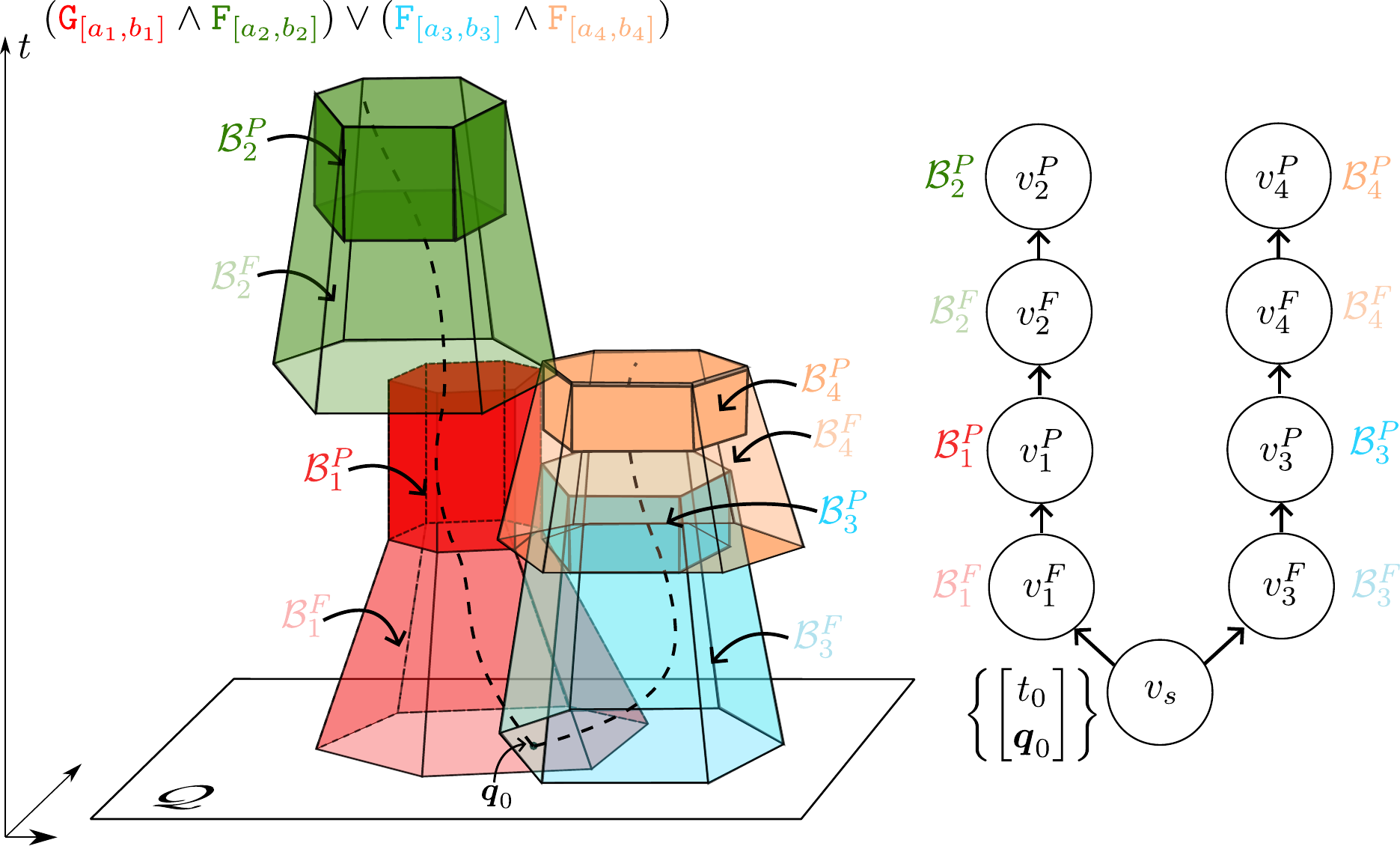}
    \caption{Illustration of the GCS associated to a simple task graph. A disjunctive STL specification induces a collection of convex spatio-temporal sets in the joint time--configuration space, including funnel sets \(\mathcal{B}_k^F\) and robustified predicate sets \(\mathcal{B}_k^P\). The two dashed trajectories show two feasible evolutions from the same initial condition \((t_0,q_0)\), each satisfying one branch of the specification by traversing the corresponding sequence of sets. The graph on the right encodes the admissible logical progression between these convex regions and provides the discrete structure underlying the trajectory-planning problem.}
    \label{fig:gsc_stl_example}
\end{figure}

A key feature of the proposed approach is that the same admissible sets used at the planning layer are also exploited at the control layer. Specifically, the planned trajectory is tracked using a CBF-based architecture that renders the active STL-induced sets forward invariant during execution. This provides a direct connection between high-level temporal-logic planning and low-level feedback control, and offers a mechanism to preserve task satisfaction even when the executed trajectory deviates from the nominal trajectory. In this sense, the framework does not treat STL satisfaction as a purely offline planning property, but rather as a closed-loop property supported by both the planner and the controller.

The main contributions of the paper are as follows:
\begin{enumerate}
    \item we introduce a GCS-based formulation for STL-constrained trajectory generation in nonconvex environments, built on a time-varying-set representation of a rich fragment of STL;
    \item we show how the same STL-induced admissible sets can be used at the control level within a forward-invariance-based feedback architecture, thereby linking planning and closed-loop task satisfaction;
    \item we validate the proposed framework in real-world experiments on a free-flyer robot and in a 3D quadrotor simulation study.
\end{enumerate}
The code will be released as open source upon acceptance.


\subsubsection*{\textbf{Related Works}}

GCS have previously been considered for temporal-logic planning in the context of LTL. In particular, \cite{kurtz2023temporal} formulates LTL motion planning as a shortest-path problem over a graph obtained from the product between a deterministic finite automaton and a transition system over a labeled partition of the configuration space. By contrast, our work considers STL, whose explicit timing semantics introduce additional complexity due to the need to satisfy tasks over prescribed time intervals. To address this, we introduce time-varying sets and their spatio-temporal counterparts, which not only encode the temporal dimension efficiently but also interface naturally with forward-invariance-based feedback control.

A complementary line of work enforces STL specifications through a combination of high-level task design and forward-invariance-based control. In \cite{yu2024continuous}, the authors introduce the signal temporal logic tree (sTLT), which converts satisfaction of nested STL formulas into a hierarchy of set-invariance conditions guiding the design of CBFs. This yields a feedback-based framework for continuous-time nonlinear systems, but it relies on reachable-set computations, which can be nontrivial in practice. Moreover, while obstacle avoidance can be incorporated through additional barrier functions, the approach is less suited to cluttered environments, where multiple obstacles may significantly increase the complexity of the underlying set computations. Finally, it does not explicitly address optimal trajectory generation.

Closer in spirit to our work, \cite{marchesini2025sampling} constructs time-varying sets whose forward invariance guarantees STL satisfaction and uses them to guide a sampling-based RRT$^\star$ planner for linear systems under input constraints. Our work shares the same invariance-based viewpoint, but replaces sampling-based search with a GCS-based shortest-path formulation over convex spatio-temporal sets. This allows us to exploit tight convex relaxations and continuous trajectory parameterizations, while naturally incorporating obstacle avoidance through the product graph construction.

Very recently and independently, \cite{chen2026signal,you2026framework} also explored GCS-based approaches for temporal-logic motion planning. In \cite{chen2026signal}, continuous-time STL motion planning is addressed through a timed-automaton abstraction coupled with a convex decomposition of the configuration space, yielding a GCS shortest-path formulation that generates smooth Bézier trajectories satisfying the specification. Our approach differs in that it works directly with time-varying convex sets whose forward invariance guarantees satisfaction with a prescribed robustness margin, and uses the same sets at the control level to maintain satisfaction during execution. Thus, while \cite{chen2026signal} focuses on motion-planning synthesis, our framework explicitly couples planning and feedback control.

The work \cite{you2026framework} considers a different class of specifications. Although phrased as a fragment of STL, the temporal operators encode untimed precedence constraints rather than explicit time intervals. This differs substantially from the timed STL setting considered here, where satisfaction depends on prescribed time windows and therefore requires reasoning in the joint time--configuration space.


\subsubsection*{\textbf{Organization}}
%
The remainder of the paper is organized as follows. Section~II-A presents the system model and the simplified planning coordinates used for trajectory synthesis. Section~II-B recalls STL and introduces the fragment considered in this work, while Section~II-C reviews the Graphs of Convex Sets framework. Section~III develops the time-varying and spatio-temporal set representation of STL formulas, together with the associated parameter-selection procedure. Section~IV introduces the task graph, collision-free graph, and product graph. Section~V formulates the resulting GCS-based trajectory-optimization problem using B-spline parameterizations. Section~VI presents the forward-invariance-based control layer. Finally, Sections~VII-A and~VII-B report experimental and simulation results, and Section~VIII concludes the paper.

\subsubsection*{\textbf{Notation}}
The sets $\nR{}$, $\nR{}_{\geq 0}$, and $\nN{}$ denote the real numbers, nonnegative real numbers, and natural numbers, respectively. Uppercase bold letters are used to denote matrices, uppercase calligraphic letters denote sets, and bold lowercase letters denote column vectors. We denote by $\bm{I}_n$ the $n\times n$ identity matrix. 
For a scalar-valued function $h : \mathbb{R}^n \to \mathbb{R}$, the \emph{$0$-superlevel set} is defined as
$\{\bm{x} \in \mathbb{R}^n \mid h(\bm{x}) \geq 0\}.$
The multiplication of a set $\mathcal{S} \subseteq \mathbb{R}^n$ by a positive scalar $\alpha \in \mathbb{R}_{>0}$ is defined as
$\alpha \mathcal{S} \define \{ \alpha \bm{x} \mid \bm{x} \in \mathcal{S} \}$.
%
%
For $t \in \mathbb{R}$ and an interval $[a,b]\subseteq \mathbb{R}$, then $t \oplus [a,b] = [t+a,\,t+b]$. For an integer \(n \in \mathbb{N}\), we let $\rangen{n}\define\{1,\dots,n\}$.
Given a differentiable scalar function $h:\mathbb{R}^n \to \mathbb{R}$ and vector fields $f:\mathbb{R}^n \to \mathbb{R}^n$ the Lie derivative of \(h\) along \(f\) is \(L_fh(\state)\define \frac{\partial h(\state)}{\partial \state}f(\state)\). A continuous function \(\alpha:[0,a)\to\Rpluseq\), with \(a>0\), belongs to class \(\mathcal{K}\) if it is strictly increasing and satisfies \(\alpha(0)=0\). A continuous function \(\alpha:(-b,a)\to\mathbb{R}\), with \(a,b>0\), belongs to extended class \(\mathcal{K}\) if it is strictly increasing and satisfies \(\alpha(0)=0\).
For an integer \(r\geq 0\), \(\mathcal{C}^r\) denotes the space of \(r\)-times continuously differentiable functions. 

Given a locally Lipschitz function $f: \nR{\xdim}\to \nR{}$, its generalized gradient at $\state \in \nR{\xdim}$ is $\partial f(\state) = \convhull\{ \lim_{i\to \infty} \nabla f(\state_i) : \state_i \to \state, \state_i \notin \mathcal{S} \union \mathcal{Q}_f \}$ where $\convhull(\mathcal{P})$ denotes the convex hull of a set $\mathcal{P}$, $\mathcal{Q}_f$ is the zero-measure set where $f$ is non-differentiable and $\mathcal{S}$ is any set of measure zero \cite{cortes2008discontinuous}.

%% file: Sections/sys_model.tex
\section{Preliminaries}
\subsection{System model}\label{sec:sys_model}

Consider the input-affine continuous-time nonlinear system
\begin{equation}\label{eq:nlsys}
    \dstate = \ff(\state) + \gG(\state)\sysinput,
\end{equation}
with state \(\state \in \stateSet \subseteq \mathbb{R}^{\xdim}\) and input \(\sysinput \in \inputSet \subseteq \mathbb{R}^{\udim}\), where \(\ff : \stateSet \to \mathbb{R}^{\xdim}\) and \(\gG : \stateSet \to \mathbb{R}^{\xdim \times \udim}\) are continuously differentiable. 

For planning purposes, we assume the availability of a planning model in chain-of-integrators form, expressed in terms of a configuration variable \(\config \in \configSet \subseteq \mathbb{R}^{\configdim}\). 
Let \(\flatoutdegree \in \mathbb{N}_{\geq 1}\) denote the order of the planning model, i.e., the number of integrators between the planning input and each component of \(\config\). Introducing the associated planning state
\begin{equation}\label{eq:planning_coordinates}
\planState
=
\begin{bmatrix}
\config^\top & \dot{\config}^\top & \cdots & \config^{(\flatoutdegree-1)\top}
\end{bmatrix}^\top
\in \mathbb{R}^{\planmodeldim},
\end{equation}
we assume that \(\planState\) is related to the original state \(\state\) through a transformation \(\planState=\feedlinT(\state)\), so that the planning dynamics are
\begin{equation}\label{eq:planning_model}
    \config^{(\flatoutdegree)} = \planInput,
\end{equation}
where \(\planInput \in \mathbb{R}^{\configdim}\) is the planning input. Such a model may arise either from a feedback-linearizing/flatness-based representation of system~\eqref{eq:nlsys} \cite{murray1995differential,fliess1995flatness} or from a suitable reduced-order model used for planning \cite{csomay2025bezier,cohen2024safety}. In general, \(\planInput\) should be interpreted as an auxiliary input of the planning model, related to the physical input \(\sysinput\) through the transformation \(\feedlinT\) and, when applicable, the corresponding feedback-linearizing input map \cite{murray1995differential}.

Under this model, any sufficiently differentiable reference trajectory
\(
\planState_d(t)=
\begin{bmatrix}
\config_d^\top(t) &
\dot{\config}_d^\top(t) &
\cdots &
\config_d^{(\flatoutdegree-1)\top}(t)
\end{bmatrix}^\top
\)
is dynamically feasible by choosing \(\planInput(t)=\config_d^{(\flatoutdegree)}(t)\). The resulting reference is then tracked on the original nonlinear system~\eqref{eq:nlsys} by the feedback controller which will be introduced in Section~\ref{sec:control}.

The configuration space \(\configSet\) contains \(\obsnum\) obstacles \(\{\obsReg_\obsID\}_{\obsID \in \rangen{\obsnum}}\), and the free space is
\begin{equation}\label{eq:freespace}
    \freespace \define \configSet \setminus \bigunion_{\obsID \in \rangen{\obsnum}} \obsReg_\obsID .
\end{equation}
Each obstacle \(\obsReg_\obsID\) is assumed to admit an implicit description through continuously differentiable functions \(\{\obsfun{\obsID,\obsconsID} : \mathbb{R}^{\configdim} \to \mathbb{R}\}_{\obsconsID \in \rangen{\obscons{\obsID}}}\). The associated obstacle function is
\begin{equation}\label{eq:obst_def}
    \obsfun{\obsID}(\config) \define \max_{\obsconsID \in \rangen{\obscons{\obsID}}} \obsfun{\obsID,\obsconsID}(\config),
\end{equation}
so that
\(
\obsReg_\obsID = \{\config \in \configSet \mid \obsfun{\obsID}(\config) < 0\},
\)
see Fig.~\ref{fig:stl_example}.

%% file: Sections/stl_intro.tex
\subsection{Signal Temporal Logic}\label{sec:stl_intro}
Signal Temporal Logic (STL) \cite{maler2004monitoring} is a formal language for specifying temporal properties of continuous-time signals. In our context, such signals are robot trajectories in configuration space. 
STL extends propositional logic with temporal operators that describe how such signals evolve over time, making it particularly suitable for expressing task specifications such as “always avoid obstacles” or “eventually reach the goal within 10 seconds.”
The basic building blocks of STL are \emph{predicates}, which capture conditions on the state space that the signal must satisfy. 
Formally, we define a \emph{predicate function} $\predFun:\nR{\configdim} \to \nR{}$, whose $0$-superlevel set is
\begin{equation}
    \predset \define \{\config \in \configSet \mid \predFun(\config)\geq 0\},
\end{equation}
which we denote as \emph{predicate set}.  
The membership condition $\config \in \predset$ is encoded by the Boolean predicate
\begin{equation}
    \pred{\predFun}(\config) = \begin{cases}
        \true, \; \; \predFun(\config) \geq 0
        \\
        \false, \; \; \predFun(\config) < 0,
    \end{cases}
\end{equation}
such that $\pred{\predFun}(\config(t)) = \true \Longleftrightarrow \config(t)\in \predset$. Starting from these atomic predicates, more complex STL formulas can be built using logical connectives (e.g., $\neg$ “not,” $\wedge$ “and”) and temporal operators. The syntax is defined recursively as
\begin{equation}\label{eq:fullstl}
\formulaT \stldef \true \stlsep \pred{h} \stlsep \neg \formulaT \stlsep  \formulaT_1 \wedge \formulaT_2 \stlsep \formulaT_1 \untilab{a}{b} \formulaT_2,
\end{equation}
where $\untilab{a}{b}$ is the until operator, expressing that $\formulaT_1$ must hold continuously until $\formulaT_2$ becomes true within the interval $\interval{a}{b}$, with $0 \leq a \leq b$.
Other operators are defined in the usual way:
\begin{itemize}
\item \emph{Disjunction}: $\formulaT_1 \vee \formulaT_2 \equiv \neg (\neg \formulaT_1 \wedge \neg \formulaT_2)$,
\item \emph{Always}: $\alwaysab{a}{b} \formulaT \equiv \neg (\true \untilab{a}{b} \neg \formulaT)$, meaning that $\formulaT$ holds at all times in the interval $\interval{a}{b}$,
\item \emph{Eventually}: $\eventuallyab{a}{b} \formulaT \equiv \true \untilab{a}{b} \formulaT$, meaning that $\formulaT$ becomes true at least once within $\interval{a}{b}$.
\end{itemize}
We here consider tasks $\formulaT$ with bounded time domain i.e. formulas whose temporal operators have a bounded domain. 

\begin{example}\label{ex:basic_stl}
Consider a mobile robot whose state is its position $\config(t) \in \mathbb{R}^2$.
We define two predicate functions:
\begin{align}
h_{\text{safe}}(\config) &= o(\config), \\
h_{\text{goal}}(\config) &= r_{\text{goal}} - \lVert \config - \pp_{\text{goal}} \rVert,
\end{align}
where the function $o(\config)$ is defined as in \eqref{eq:obst_def}, 
$\pp_{\text{goal}} \in \mathbb{R}^2$ is the goal position, and $r_{\text{goal}}>0$ is the goal tolerance radius (see Fig.~\ref{fig:stl_example}).
The predicate $\pred{h_{\text{safe}}}(\config(t))$ is true, at time $t$, whenever the robot is outside the obstacle region.
The predicate $\pred{h_{\text{goal}}}(\config(t))$ is true whenever the robot lies within the goal region.
The $0$-superlevel set of $h_{\text{safe}}(\config)$ and $h_{\text{goal}}(\config)$ are respectively indicated as $\mathcal{H}_{\text{safe}}$ and $\mathcal{H}_{\text{goal}}$.
The formula $\phi = \alwaysab{0}{T} \pred{h_{\text{safe}}}
    \wedge
    \eventuallyab{0}{10}\pred{h_{\text{goal}}}$
formalizes the requirement: “the robot must always remain safe during the horizon $[0,T]$ and must eventually reach the goal within 10 seconds.” 

\begin{figure}
    \centering
    \includegraphics[width=0.6\linewidth]{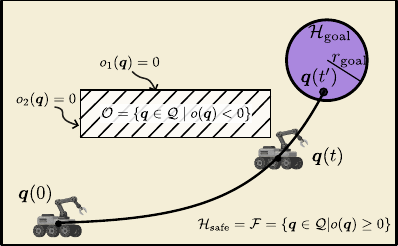}
    \caption{Illustration of Example~\ref{ex:basic_stl}. The robot satisfies the STL formula if and only if $t'\in \interval{0}{10}$. 
    }
    \label{fig:stl_example}
\end{figure}
\end{example}

To characterize satisfaction of an STL formula by a given signal, we consider quantitative (\emph{robust}) semantics, which assigns to each formula a real-valued robustness measure $\robustness{\formulaT}(\config,t) \in \mathbb{R}$. 
This robustness characterizes not only whether a signal satisfies a given STL task, but also \emph{how robustly} it does so. Formally, let $(\config,t) \stlsatisfies \formulaT$ denote that the signal $\config$, starting from time $t$, satisfies $\formulaT$. 
The robustness $\robustness{\formulaT}(\config,t)$ is recursively defined as:
\begin{subequations}\label{eq:robust semantics}
\begin{align}
\robustness{\pred{}}(\config,t)&=\predFun(\config(t)), \label{eq:predicate robust}\\
\robustness{\neg \psi}(\config,t)& = -\robustness{\psi}(\config,t), \label{eq:negation robust}\\
\robustness{F_{[a, b]} \psi}(\config,t)&=\max _{\tau \in t\oplus [a, b]} \{\robustness\psi\left(\config, \tau \right)\}, \label{eq:eventually robust}\\
\robustness{G_{[a, b]} \psi}(\config,t)&= \min _{\tau \in t\oplus[a, b]} \{\robustness\psi\left(\config, \tau\right)\},\label{eq:always robust}\\
\robustness{\psi_1 U_{[a,b]}\psi_2}(\config,t) &= \max_{\tau \in t\oplus[a,b]}\label{eq:until robust} \\ 
& \hspace{-1.5cm} \{ \min\{\robustness{\psi_2}(\config, \tau ), \min_{\tau' \in [t,\tau]} \robustness{\psi_1}(\config,\tau')\}\},\notag\\
\robustness{\psi_1 \wedge \psi_2}(\config,t)& = \min \left\{\robustness{\psi_1}(\config,t), \robustness{\psi_2}(\config,t)\right\},  \label{eq:conjunction robust} 
\end{align}
\end{subequations}
From \eqref{eq:robust semantics}, it follows that  $\robustness{\formulaT}(\config,t)>0 \implies (\config,t) \stlsatisfies \formulaT$ \cite[Prop. 16]{yu2024continuous}. 
Here, the time $t$ indicates the starting point at which we evaluate satisfaction of the STL formula with respect to the entire trajectory. 
Unlike other semantics, the quantitative semantics provide a \emph{degree of satisfaction}. 
In particular, we say that $\config$ \emph{robustly satisfies} $\formulaT$ with margin $\robustMargin>0$ if $\robustness{\formulaT}(\config,t) \geq \robustMargin$.


\subsection*{STL Fragment for Graph-of-Convex-Sets Planning}\label{subsec:stl_frag}

The problem of synthesizing dynamically feasible trajectories that satisfy a given STL specification (hereafter \emph{STL synthesis}) under the complete STL grammar \eqref{eq:fullstl} is known to be NP-hard, and thus, often computationally intractable. In this work, we restrict attention to a broad class of STL specifications that remains practically relevant while enabling a more tractable treatment. 

First, we assume that, for each predicate function \(\predFun(\config)\) defined over the configuration variable \(\config\), the corresponding \(0\)-superlevel set 
\(
    \predset
\)
is convex. A typical and practically relevant case is that of polyhedral predicates. In particular, given \(n_h>0\), \(\cc_k\in\nR{\configdim}\), and \(d_k\in\nR{}\), we define
%
\begin{equation}
h(\config) = \min_{k \in [n_{h}]} \{-\cc_k^\top \config + d_k\},
\label{eq:pred_set}
\end{equation}
Then the set \(\predset\) is the polyhedron
$\predset = \{\config \mid \cC \config \leq \dd \}$ with $\cC = [\cc_1 \ \dots \ \cc_{\numconsh}]^\top \in \mathbb{R}^{n_h \times \configdim}$ and $\dd \define \begin{bmatrix}
    d_1 & \hdots & d_{\numconsh}
\end{bmatrix}^\top \in \mathbb{R}^{n_h}$.

Second, we consider a subfragment of STL, 
which we refer to as Disjunctive Convex Signal Temporal Logic (dc-STL), in analogy to the Disjunctive Normal Form considered in first-order logic:
\begin{subequations}
    \begin{align}
    \formulaNT &\stldef  
    \tempopab{a}{b}\pred{h} \stlsep \tempopab{a}{b}\tempopab{a'}{b'}'\pred{h} \stlsep \pred{h_1} \untilab{a}{b} \pred{h_2},
    \label{eq:temp_extended_pred}
    \\
    \formulaT &\stldef \formulaNT \mid  \formulaT_1 \wedge \formulaT_2,
    \label{eq:stl_conj_phi}
    \\
    \formulaOR & \stldef \formulaT \stlsep \formulaT_1 \vee \formulaT_2, 
    \label{eq:stl_or_formulas}
\end{align}%
\label{eq:stlfrag}
\end{subequations}%
where $\tempop , \tempop' \in \{\eventually, \always\}$.
Formulas of class \(\formulaNT\) represent elementary temporal requirements involving either a single temporal operator or a composition of two, e.g., “eventually reach a region,” “always eventually monitor an area,” or “recurrently visit a location of interest.” More complex behaviors are obtained by conjunction, yielding formulas of class \(\formulaT\) as in \eqref{eq:stl_conj_phi} (see Example~\ref{ex:basic_stl}). Finally, formulas of class \(\formulaOR\), as in \eqref{eq:stl_or_formulas}, capture disjunctive requirements, i.e., alternative conjunctions of temporal requirements, and are satisfied whenever at least one branch is satisfied.
The separation between conjunctions and disjunctions is introduced only to simplify the exposition. Moreover, any Boolean combination of formulas of class \(\formulaNT\) using only \(\wedge\) and \(\vee\) can be rewritten in the form \eqref{eq:stlfrag} by distributing conjunctions over disjunctions, as in the standard conversion to disjunctive normal form.


From a trajectory-optimization perspective, the \emph{feasible region} associated with fragment \eqref{eq:stlfrag} is the set of all trajectories, or equivalently all decision variables of the STL synthesis problem, that satisfy a formula in the fragment. This feasible region can be interpreted as the logical OR of multiple convex sets, each corresponding to a distinct mode of satisfaction, whose union is in general nonconvex \cite{ceria1999convex}. It is precisely this structure that makes the considered STL fragment amenable to a Graphs of Convex Sets formulation, as shown in the sequel.

\subsection{Graphs of Convex Sets formulation}\label{sec:gcs_form}


GCS provide a convenient framework for representing trajectory-planning problems over nonconvex domains through a graph of convex regions \cite{marcucci2024shortest,marcucci2023motion}. 

Formally, a GCS is a directed graph
$
\graph_{\text{GCS}} = (\vertices_{\text{GCS}}, \edges_{\text{GCS}}),
$
where each vertex \(v \in \vertices_{\text{GCS}}\) is associated with a decision variable \(\vertexvar{v} \in \mathbb{R}^{n_v}\), a closed convex set \(\vertexset{v} \subseteq \mathbb{R}^{n_v}\), and a convex cost \(\ell_v\), while each edge \(e=(u,v)\in\edges_{\text{GCS}}\) is associated with a closed convex set \(\vertexset{e} \subseteq \mathbb{R}^{n_u+n_v}\) and a convex cost \(\ell_e\) \cite{marcucci2024shortest, marcucci2025unified}. 

Given a source vertex and a target vertex, let \(\pathsSet\) denote the set of all source--target paths in \(\graph_{\text{GCS}}\), and let \(\edges_{\path}\subseteq\edges_{\text{GCS}}\) denote the edges traversed by a path \(\path\in\pathsSet\). The corresponding GCS shortest-path problem is
\begin{subequations}\label{eq:gcs_spp}
\begin{align}
\min_{\substack{\path \in \pathsSet \\ \vertexvar{v},\, v\in\vertices_{\text{GCS}}}}
\quad &
\sum_{v\in\vertices_{\text{GCS}}} \ell_v(\vertexvar{v})
+
\sum_{(u,v)\in\edges_{\path}} \ell_e(\vertexvar{u},\vertexvar{v})
\label{eq:gcs_spp_a}
\\
\text{s.t.}\quad
& \vertexvar{v} \in \vertexset{v},
\qquad \forall v \in \path,
\label{eq:gcs_spp_c}
\\
& (\vertexvar{u},\vertexvar{v}) \in \vertexset{e},
\qquad \forall (u,v)\in\edges_{\path}.
\label{eq:gcs_spp_d}
\end{align}
\end{subequations}
The specific form of the vertex variables, costs, and coupling constraints depends on the trajectory parameterization and on the geometric structure of the problem at hand. In the present work, the vertex variables \(\vertexvar{v}\) will later parameterize the control points of B-spline \cite{biagiotti2008trajectory} segments associated with the graph vertices, while the costs \(\ell_v\) and \(\ell_e\) will encode surrogates of trajectory length and the constraint sets \(\vertexset{v}\) and \(\vertexset{e}\) will enforce smoothness and STL-related conditions on the resulting trajectory. In particular, in Section~\ref{sec:stl2tvsets} we show how STL constraints can be reformulated as convex sets in the joint time--configuration space, which will later define the GCS constraint sets \(\vertexset{v}\) and \(\vertexset{e}\) (see Fig.~\ref{fig:gsc_stl_example}).



%% file: Sections/stl_tv_sets.tex
\section{From STL tasks to time-varying sets}\label{sec:stl2tvsets}
Prior work \cite{lindemann2017prescribed, lindemann2018control, charitidou2021barrier, marchesini2025sampling} has shown that the \textit{STL synthesis} problem can be recast as the problem of maintaining the system state within a suitable \emph{time-varying set}. Namely, given a formula $\formulaOR$ in fragment \eqref{eq:stlfrag}, it is possible to algorithmically construct a time varying set $\tvset{}_\formulaOR(t)$ with the following property:
\begin{equation}\label{eq:sufficient condition}
    \forall t \in \Rpluseq, \quad \config(t) \in \tvset{}_ \formulaOR(t) \quad \implies \quad (\config, 0) \stlsatisfies \formulaOR.
\end{equation}
In other words, condition \eqref{eq:sufficient condition} replaces the satisfaction condition of the formula $\formulaOR$ with a sufficient condition in terms of the \textit{forward invariance} of a suitably defined time-varying set:
\begin{definition}[Forward Invariance]\label{def:fw}
A time-varying set $\tvset{}(t)$ is \emph{forward invariant} for system \eqref{eq:planning_model} from the initial condition $\config_0 :=\config(t_0) \in \tvset{}(t_0)$, $t_0 \geq 0$, if there exists a control input $\sysinput : [t_0,\infty) \rightarrow \inputSet$ such that 
$$\config(t) \in \tvset{}(t), \qquad \forall t \ge t_0,$$
holds along the corresponding solution of \eqref{eq:planning_model}.
\end{definition}
Although the forward invariance condition \eqref{eq:sufficient condition} is only sufficient, it allows the high-level STL synthesis problem to be reformulated as a constrained trajectory-planning problem.

The remainder of the section is organized as follows. Section~\ref{subsec:tvfun} constructs the set \(\tvset{}_{\formulaNT}(t)\) associated with a single formula \(\formulaNT\) in \eqref{eq:temp_extended_pred}, so that the sufficient condition \eqref{eq:sufficient condition} holds. Sections~\ref{subsec:compos_sets} and~\ref{sec:specifications} then show how to recursively construct the sets \(\tvset{}_{\formulaT}(t)\) and \(\tvset{}_{\formulaOR}(t)\) associated with formulas \(\formulaT\) and \(\formulaOR\), by taking intersections and unions of sets of the form \(\tvset{}_{\formulaNT}(t)\) according to the grammar in \eqref{eq:stlfrag}.

\subsection{Single Task To Time-Varying Sets}\label{subsec:tvfun}
Given a single STL task \(\formulaNT_\taskID\), with predicate function \(\predFun_\taskID\) as in \eqref{eq:temp_extended_pred}, we seek to construct a time-varying set \(\tvset{}_\taskID(t)\) satisfying \eqref{eq:sufficient condition}. To simplify notation, we write \(\tvset{}_\taskID(t)\) instead of \(\tvset{}_{\formulaNT_\taskID}(t)\). We restrict the discussion to formulas of the form \(\formulaNT_\taskID=\alwaysab{a}{b}\mu^h\) and \(\formulaNT_\taskID=\eventuallyab{a}{b}\mu^h\). Nested formulas can be handled as described in Remark~\ref{remark:nested}.

Starting from a given initial configuration \(\config_0\), the key idea is to construct, for each formula \(\formulaNT_\taskID\), a time-varying set \(\tvset{}_\taskID(t)\) that initially contains \(\config_0\) and then shrinks over time until it lies inside the predicate set associated with \(\formulaNT_\taskID\), with the timing required by the formula. Formally, we define
\begin{equation}
    \tvset{}_\taskID(t) \define \{\config \in \configSet \mid \stlcbf{}_\taskID(\config, t) \geq 0\},
    \quad t \in \interval{\tinittask{\taskID}}{\tfinaltask{\taskID}},
\end{equation}
where 
\begin{equation}\label{eq:main barrier function}
    \stlcbf{}_\taskID(\config, t) \define \predFun_\taskID(\config) + \lintermtv_\taskID(t).
\end{equation}
The function $\lintermtv_\taskID(t)$ in \eqref{eq:main barrier function} is a piecewise linear monotonically decreasing function,  with switching times $(\tinittask{\taskID}, \alpha_\taskID, \tfinaltask{\taskID})$ (Fig.~\ref{fig:gamma general image}):
\begin{equation}\label{eq:kappa_def}
    \lintermtv_\taskID(t) =
    \begin{cases}
         -\tfrac{\bar{\lintermtv}_{\taskID}}{\alpha_{\taskID}} t + \bar{\lintermtv}_\taskID - \robustMargin, & t \in [\tinittask{\taskID}, \alpha_{\taskID}) \\
         -\robustMargin, & t \in [\alpha_{\taskID}, \tfinaltask{\taskID}],
    \end{cases}
\end{equation}
where constant $\bar{\lintermtv}_{\taskID} > 0$ determines the slope of the decreasing transient. Let us define the robustified predicate set as
\begin{equation}
\robpredset_{\taskID}
\define
\{\config \in \configSet \mid \predFun_{\taskID}(\config)-\robustMargin \geq 0\}.
\end{equation}
The role of \(\lintermtv_\taskID(t)\) is to create a \emph{funnel} that progressively drives \(\tvset{}_\taskID(t)\) toward \(\robpredset_{\taskID}\), as illustrated in Fig.~\ref{fig:gamma general image}. Indeed, for every \(t \in [\alpha_\taskID,\tfinaltask{\taskID}]\), one has \(\lintermtv_\taskID(t)=-\robustMargin\). Hence, by \eqref{eq:main barrier function},
\[
\stlcbf{}_\taskID(\config,t)\geq 0
\quad\Longleftrightarrow\quad
\predFun_\taskID(\config)-\robustMargin \geq 0,
\]
which implies
\[
\tvset{}_\taskID(t)=\robpredset_{\taskID}\subseteq \predset_{\taskID},
\qquad
\forall t\in[\alpha_\taskID,\tfinaltask{\taskID}].
\]

This construction yields the desired sufficient condition for formulas involving either the \(\alwaysab{a_\taskID}{b_\taskID}\) or the \(\eventuallyab{a_\taskID}{b_\taskID}\) operator through an appropriate choice of \(\alpha_\taskID\) and \(\tfinaltask{\taskID}\). In particular, for \(\alwaysab{a_\taskID}{b_\taskID}\), we choose \(\alpha_\taskID=a_\taskID\) and \(\tfinaltask{\taskID}=b_\taskID\), so that \(\tvset{}_\taskID(t)\subseteq \predset_{\taskID}\) for all \(t\in[a_\taskID,b_\taskID]\). For \(\eventuallyab{a_\taskID}{b_\taskID}\), we instead choose \(\alpha_\taskID=\tfinaltask{\taskID}=t'\), with \(t'\in[a_\taskID,b_\taskID]\), so that the set reaches \(\robpredset_{\taskID}\) at time \(t'\), which is sufficient to enforce \(\eventuallyab{a_\taskID}{b_\taskID}\pred{\predFun_\taskID}\) (see Fig.\ref{fig:gamma general image}). 

Viewed in the joint time--configuration space, the set \(\tvset{}_\taskID(t)\) can be interpreted as a spatio-temporal funnel within which trajectories must evolve to satisfy \(\formulaNT_\taskID\) with robustness margin \(\robustMargin\). Accordingly, letting \(\spatiotemp \define (t,\config) \in \Rpluseq \times \configSet\), we define the associated spatio-temporal set
\begin{equation}\label{eq:spatiotempset}
    \tvsetspatiotemp_{\taskID}
    \define
    \left\{
    \spatiotemp \in \Rpluseq \times \configSet
    \;\middle|\;
    \stlcbf{}_\taskID(\config,t) \geq 0,\;
    t \in \interval{\tinittask{\taskID}}{\tfinaltask{\taskID}}
    \right\},
\end{equation}
so that \(\tvset{}_\taskID(t)\) is recovered as the slice of \(\tvsetspatiotemp_{\taskID}\) at time \(t\). Moreover, the geometry of \(\tvsetspatiotemp_{\taskID}\) naturally decomposes into two convex components:

1) the \emph{funnel spatio-temporal set}
\begin{equation}\label{eq:funnelset}
    \funnelset{\taskID}
    =
    \left\{
    \spatiotemp \in \Rpluseq \times \configSet
    \;\middle|\;
    \stlcbf{}_\taskID(\config,t) \geq 0,\;
    t \in \interval{\tinittask{\taskID}}{\alpha_\taskID}
    \right\},
\end{equation}
which shrinks over time toward the target region;

2) the \emph{robustified predicate spatio-temporal set}
\begin{equation}\label{eq:predicateset}
    \predicateset{\taskID}
    =
    \left\{
    \spatiotemp \in \Rpluseq \times \configSet
    \;\middle|\;
    \config \in \robpredset_{\taskID},\;
    t \in \interval{a_\taskID}{\tfinaltask{\taskID}}
    \right\},
\end{equation}
which corresponds to the phase in which the trajectory is required to satisfy the predicate with robustness margin \(\robustMargin\) (see the left panels of Fig.~\ref{fig:gamma general image}). Therefore, the spatio-temporal task set is equivalently represented as the union of these two convex components,
\[
\tvsetspatiotemp_{\taskID} = \funnelset{\taskID} \cup \predicateset{\taskID}.
\]

\begin{figure}
    \centering
    \includegraphics[width=\linewidth]{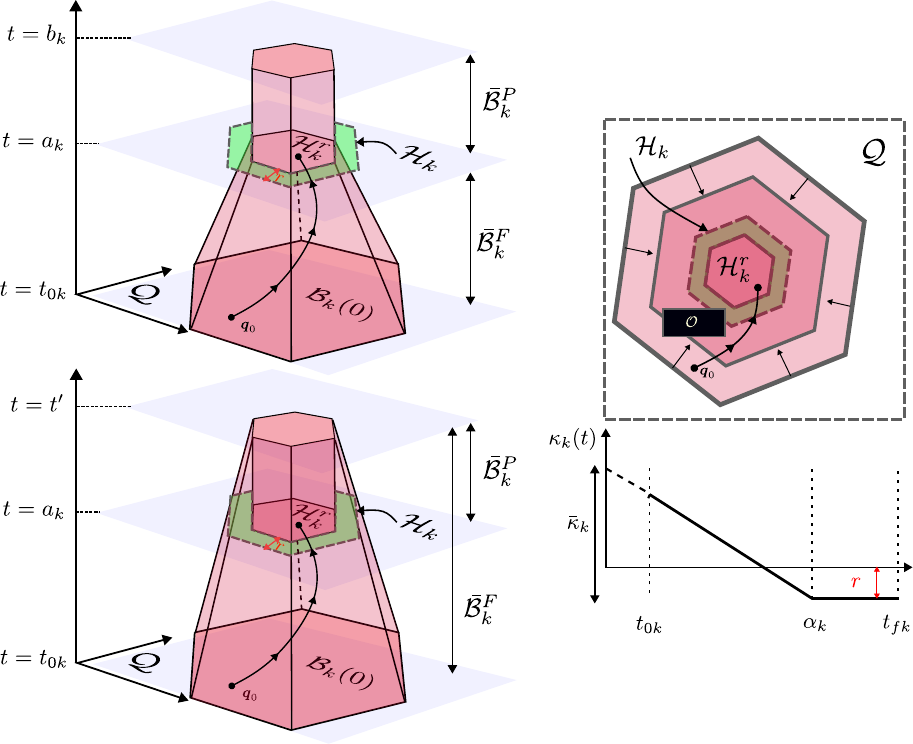}
    \caption{Set representation of a formula \(\formulaNT_\taskID\) induced by \eqref{eq:main barrier function} and \eqref{eq:kappa_def} for linear predicates. Left: spatio-temporal set \(\tvsetspatiotemp_{\taskID}\) for a \(\always\) operator (top) and an \(\eventually\) operator (bottom). Top right: evolution of \(\tvset{}_\taskID(t)\) in the configuration space, shrinking toward the robustified predicate set \(\robpredset_{\taskID}\). Bottom right: piecewise-linear function \(\lintermtv_\taskID(t)\).}
    \label{fig:gamma general image}
\end{figure}



\subsection{Conjunctions to Time-Varying Sets}\label{subsec:compos_sets}

Having defined the construction associated with a formula of class $\formulaNT$, we now show how to construct the time-varying set \(\tvset_{\formulaT}(t)\) associated with a conjunction \(\formulaT\). The key idea is that the set $\tvset{}_{\formulaT}(t)$ is obtained by intersecting the task sets corresponding to the tasks that are active at a given time. To make this construction precise, one must first specify how the tasks in the conjunction are ordered and, consequently, over which time intervals their associated sets are enforced.

Consider a conjunction of tasks
$
\formulaT = \bigwedge_{\taskID=1}^{\tasknum} \formulaNT_\taskID,
$
where each task \(\formulaNT_\taskID\) is associated with an STL interval \([a_\taskID,b_\taskID]\). We define the undirected \emph{overlap graph} \(\graph_{\mathrm{ov}}=(\vertices_{\mathrm{ov}},\edges_{\mathrm{ov}})\), with \(\vertices_{\mathrm{ov}}=\{1,\dots,\tasknum\}\), by
\[
(\taskID_1,\taskID_2)\in\edges_{\mathrm{ov}}
\quad\Longleftrightarrow\quad
\interval{a_{\taskID 1}}{b_{\taskID 1}}
\cap
\interval{a_{\taskID 2}}{b_{\taskID 2}}
\neq
\emptyset.
\]
Thus, an edge connects two tasks whenever their STL intervals overlap. The tasks that must be handled jointly are then exactly the connected components of \(\graph_{\mathrm{ov}}\). Since the connected components of a graph are uniquely defined, this yields a unique partition of the task set into $m_\phi$ groups (see Fig.~\ref{fig:overlapping}), which we denote by
\[
\{\mathcal{K}_{\mathrm{ov}}^1,\mathcal{K}_{\mathrm{ov}}^2,\ldots,\mathcal{K}_{\mathrm{ov}}^{m_\phi}\},
\qquad m_\phi \ge 1.
\]

Equivalently, two tasks \(\taskID_1\) and \(\taskID_2\) belong to the same group if and only if there exists a sequence of tasks
\[
\taskID_1=\ell_1,\ell_2,\dots,\ell_m=\taskID_2
\]
such that
\[
\interval{a_{\ell_j}}{b_{\ell_j}}
\cap
\interval{a_{\ell_{j+1}}}{b_{\ell_{j+1}}}
\neq
\emptyset,
\qquad j\in \rangen{m-1}.
\]
\begin{figure}[b]
    \centering
    \includegraphics[width=\linewidth]{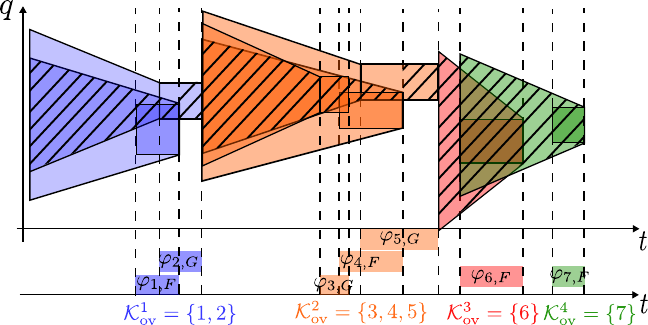}
    \caption{Illustration of the grouping of tasks in the conjunction \(\formulaT=\bigwedge_{k=1}^{N}\formulaNT_k\), shown in a one-dimensional configuration space. The striped region represents the admissible trajectory evolution. For brevity, the operator type is indicated in the subscript, e.g., \(\formulaNT_{1,\eventually{}{}}\).}
    \label{fig:overlapping}
\end{figure}

The simplest group is therefore a singleton, corresponding to a task whose STL interval does not overlap with any other one. More generally, each group contains tasks whose intervals are linked, directly or transitively, by overlap, and which therefore cannot be arranged in a strict sequential order internally. 

The groups themselves can instead be ordered uniquely by their temporal location (see Fig.~\ref{fig:overlapping}). Indeed, if \(\mathcal{K}_{\mathrm{ov}}^i\) and \(\mathcal{K}_{\mathrm{ov}}^j\) are two distinct connected components of \(\graph_{\mathrm{ov}}\), then no interval in one component overlaps any interval in the other. Hence, the groups are temporally ordered, i.e., for any \(i\neq j\),
\[
\max_{\taskID \in \mathcal{K}_{\mathrm{ov}}^i} b_\taskID
<
\min_{\taskID \in \mathcal{K}_{\mathrm{ov}}^j} a_\taskID
\quad\text{or}\quad
\max_{\taskID \in \mathcal{K}_{\mathrm{ov}}^j} b_\taskID
<
\min_{\taskID \in \mathcal{K}_{\mathrm{ov}}^i} a_\taskID.
\]
Accordingly, after relabeling the components if necessary, we assume they are indexed so that
\[
\max_{\taskID \in \mathcal{K}_{\mathrm{ov}}^i} b_\taskID
<
\min_{\taskID \in \mathcal{K}_{\mathrm{ov}}^{i+1}} a_\taskID,
\qquad
i \in \rangen{m_\phi-1} 
\]

Once the groups have been formed, each task \(\formulaNT_\taskID\), with \(\taskID \in \mathcal{K}_{\mathrm{ov}}^i\), is associated with an interval
$
[\tinittask{\taskID},\,\tfinaltask{\taskID}],
$
over which the corresponding time-varying set is defined. The final time is fixed as
$
\tfinaltask{\taskID}=b_{\taskID},
$
whereas the initial time \(\tinittask{\taskID}\) depends on the type of the tasks contained in the preceding group \(\mathcal{K}_{\mathrm{ov}}^{i-1}\) and is specified in Section~\ref{sec:param_selection}.

The activation intervals are chosen so as to be consistent with the ordering of the groups. At any given time, the tasks that are actually enforced belong to a unique group, although tasks in later groups may already be defined. We therefore define the \emph{active set of tasks} at time \(t\) as
\begin{equation}
\mathcal{I}^{\formulaT}(t) \subseteq \rangen{\tasknum} 
\label{eq:active_set}
\end{equation}
where \(\mathcal{I}^{\formulaT}(t)\) contains precisely the tasks whose corresponding time-varying sets are currently enforced at time \(t\). In particular,
\[
\taskID \in \mathcal{I}^{\formulaT}(t)
\;\Longrightarrow\;
\tinittask{\taskID} \le t \le \tfinaltask{\taskID},
\]
but the converse does not necessarily hold. Indeed, a task may be defined over a given interval without being active yet, since activation must also be consistent with the completion of the preceding group in the chosen ordering. For tasks involving the \eventuallytxt{} operator, moreover, termination is not determined solely by time.
More precisely, tasks of type \alwaystxt{} remain active until the deterministic time \(t=b_{\taskID}\), whereas tasks of type \eventuallytxt{} remain active only until they are satisfied, namely until the trajectory enters the corresponding robustified predicate set at some time
$
t' \in \interval{a_\taskID}{b_\taskID}.
$
Once this happens, the corresponding eventually task is regarded as completed, and its associated constraint is no longer enforced. 

The same holds for formulas involving the \untiltxt{} operator. In particular, a formula
\(
    \pred{h_{\taskID 1}} \untilab{a}{b} \pred{h_{\taskID 2}}
\)
is treated as an \eventuallytxt{}-type task whose funnel phase additionally enforces the first predicate. Let \(\funnelset{\taskID 2}\) and \(\predicateset{\taskID 2}\) denote, respectively, the funnel and predicate spatio-temporal components associated with the \eventuallytxt{} task for \(h_{\taskID 2}\) over \(\interval{a}{b}\). Then the corresponding \untiltxt{} components are defined as
\begin{equation}
    \funnelset{\taskID 1\untilab{a}{b}\taskID 2}
    \define
    \predicateset{\taskID 1}
    \cap
    \funnelset{\taskID 2},
    \label{eq:until_funnel_set}
\end{equation}
while the predicate component is defined as
\(
    \funnelset{\taskID 1\untilab{a}{b}\taskID 2}
    \define
    \predicateset{\taskID 1} \cap \funnelset{\taskID 2}.
\)
Thus, during the funnel phase, the trajectory is required to satisfy \(h_{\taskID 1}\) while being driven toward the robustified predicate set of \(h_{\taskID 2}\). Once the latter predicate component is reached, the \untiltxt{} formula is completed, exactly as for an \eventuallytxt{} task.

The set corresponding to the active set of tasks \(\mathcal{I}^{\formulaT}(t)\) is then
\begin{equation}
    \tvset{}_\formulaT(t)
    =
    \bigcap_{\taskID \in \mathcal{I}^{\formulaT}(t)} \tvset{}_\taskID(t),
    \label{eq:global_plan_set}
\end{equation}
and the associated spatio-temporal set is
\begin{equation}
    \tvsetspatiotemp_\formulaT
    =
    \left\{
    \spatiotemp=(t,\config) \in \Rpluseq \times \configSet
    \;\middle|\;
    \config \in \tvset{}_\formulaT(t)
    \right\}.
    \label{eq:st_global_plan_set_expanded}
\end{equation}

As an example, consider the simple conjunction
\[
\formulaT = \formulaNT_{\taskID 1} \wedge \formulaNT_{\taskID 2},
\]
with tasks $\formulaNT_{\taskID 1}
=
\tempopab{a_{\taskID 1}}{b_{\taskID 1}}\pred{h_{\taskID 1}}$ and
$\formulaNT_{\taskID 2}
=
\tempopab{a_{\taskID 2}}{b_{\taskID 2}}\pred{h_{\taskID 2}}.
$
We distinguish two cases according to the relation between the STL intervals \([a_{\taskID 1},b_{\taskID 1}]\) and \([a_{\taskID 2},b_{\taskID 2}]\).

\emph{(i) Sequential tasks.}
If
\[
\interval{a_{\taskID 1}}{b_{\taskID 1}}
\cap
\interval{a_{\taskID 2}}{b_{\taskID 2}}
=
\emptyset,
\]
then the two STL tasks can be ordered unambiguously. In this case, the conjunction is represented by sequencing the corresponding sets in time (Fig.~\ref{fig:overlapping}). More precisely, the activation of the second task is chosen consistently with the completion of the first one, so that the conjunction set is first determined by \(\tvset{}_{\taskID 1}(t)\) and then by \(\tvset{}_{\taskID 2}(t)\). The precise transition between the two tasks is enforced by the graph construction introduced later in Section~\ref{sec:param_selection}.

\emph{(ii) Overlapping tasks.}
If
\[
\interval{a_{\taskID 1}}{b_{\taskID 1}}
\cap
\interval{a_{\taskID 2}}{b_{\taskID 2}}
\neq \emptyset,
\]
then the two STL tasks cannot be ordered unambiguously and must be enforced simultaneously over the time interval in which they are both active. In that interval, the corresponding conjunction set is the intersection of the two task sets (Fig.~\ref{fig:overlapping}):
\begin{equation}
\begin{aligned}
\tvset{}_\formulaT(t)
&=
\tvset{}_{\taskID 1}(t) \cap \tvset{}_{\taskID 2}(t)
\\
&=
\left\{
\config \in \configSet
\;\middle|\;
\min\!\bigl(
\stlcbf{}_{\taskID 1}(\config,t),
\stlcbf{}_{\taskID 2}(\config,t)
\bigr)\ge 0
,\right.
\\
& \hspace{4cm} \left. \taskID_1,\taskID_2 \in \mathcal{I}^{\formulaT}(t)\right\}.
\end{aligned}
\label{eq:overlap_intersection}
\end{equation}
If one of the two tasks terminates before the other, then it is removed from the active set \(\mathcal{I}^{\formulaT}(t)\). From that time onward, the corresponding time-varying set is no longer enforced, and the conjunction set reduces to the task set induced by the remaining active task.

Let
\[
\mathcal{T}_s
\define
\{\,t\in\Rpluseq \mid \mathcal{I}^{\formulaT}(t^-)\neq \mathcal{I}^{\formulaT}(t^+)\,\}
\]
denote the set of \emph{switching times}, namely the times at which the active-task set changes. 

\begin{definition}[Set expansion condition at switching times]\label{def:setexpansion}
The conjunction-induced time-varying set \(\tvset{}_\formulaT(t)\) is said to satisfy the \emph{set expansion condition at switching times} if, for every switching time \(t_s\in\mathcal{T}_s\),
\begin{equation}
    \lim_{t \to t_s^-} \tvset{}_\formulaT(t)
    \subseteq
    \lim_{t \to t_s^+} \tvset{}_\formulaT(t).
    \label{eq:expansion_condition}
\end{equation}
Equivalently, for every \(t_s\in\mathcal{T}_s\),
\[
\config \in \lim_{t \to t_s^-} \tvset{}_\formulaT(t)
\;\Longrightarrow\;
\config \in \lim_{t \to t_s^+} \tvset{}_\formulaT(t).
\]
\end{definition}

Condition~\eqref{eq:expansion_condition} ensures that a configuration that is admissible immediately before a switching time remains admissible immediately after it. Although this condition is not strictly necessary for open-loop planning, it becomes essential when designing a feedback law that renders \(\tvset{}_\formulaT(t)\) forward invariant. In practice, it is enforced through the choice of the parameters \(\bar{\lintermtv}_{\taskID}\), whose selection is discussed in Section~\ref{sec:param_selection}, which regulate the shrinking of the funnels associated with the active tasks.

\begin{remark}\label{remark:nested}
Nested formulas can be treated by combining the constructions developed in this and the next section. In particular, \(\eventuallyab{a}{b}\alwaysab{a'}{b'}\mu^h\) can be interpreted as a time-shifted always requirement, whereas \(\alwaysab{a}{b}\eventuallyab{a'}{b'}\mu^h\) can be viewed as a conjunction of \eventuallytxt{} requirements over suitable intervals, and is therefore handled using the conjunction construction introduced in the next section. A complete treatment is given in \cite[Sec.~IV.C]{marchesini2025sampling}.
\end{remark}

\subsubsection{\underline{Parameter selection and task sequencing}}\label{sec:param_selection}

So far, we have not yet addressed the selection of the parameter \(\tinittask{\taskID}\), which determines when the set associated with \(\formulaNT_\taskID\) becomes active, nor the selection of the slope \(\bar{\kappa}_{\taskID}\), which determines the evolution of the corresponding funnel set.

The choice of \(\tinittask{\taskID}\) is made consistently with the sequential ordering of the groups introduced in Section~\ref{subsec:compos_sets}. In particular, all tasks belonging to the same group \(\mathcal{K}_{\mathrm{ov}}^i\) are assigned the same activation time, denoted by \(\tinittask{i}\), so that they become active simultaneously. 

The value of \(\tinittask{i}\) is selected consistently with the completion of the preceding group \(\mathcal{K}_{\mathrm{ov}}^{i-1}\). For every group \(\mathcal{K}_{\mathrm{ov}}^i\), with \(i>1\), the next group may become active only after all formulas in \(\mathcal{K}_{\mathrm{ov}}^{i-1}\) have completed. Accordingly, we choose
\begin{equation}
\tinittask{i}
=
\max\!\left(
\max_{\ell\in\mathcal{K}_{\mathrm{ov}}^{i-1}\,:\,\formulaNT_\ell\text{ of type }\always} b_\ell,\;
\max_{\ell\in\mathcal{K}_{\mathrm{ov}}^{i-1}\,:\,\formulaNT_\ell\text{ of type }\eventually} a_\ell
\right).
\end{equation}
Thus, all tasks in \(\mathcal{K}_{\mathrm{ov}}^i\) inherit the same activation time \(\tinittask{i}\).

The slope parameters \(\bar{\kappa}_{\taskID}\) are instead chosen so that, at every switching time \(t_s\) of the active-task set \(\mathcal{I}^{\formulaT}(t)\), the inclusion in Definition~\ref{def:setexpansion} holds. Since \(\tvset{}_\formulaT(t)\) is convex at each switching time, this containment condition can be enforced efficiently via convex optimization. For simplicity, we do not detail here the construction of the sets associated with overlapping formulas. Appendix~\ref{app:computingkappa} presents one possible approach for designing \(\bar{\lintermtv}_{\taskID+1}\) for polytopic sets associated with single formulas so that \eqref{eq:expansion_condition} holds, together with possible extensions to more general convex sets. The resulting design problem is closely related to the \emph{polytope containment problem}~\cite{mangasarian2002set,sadraddini2019linear}. See also \cite[Sec.~V.C]{marchesini2025sampling} for a possible design of the parameters for overlapping tasks.


\subsection{Disjunctions to Time-Varying Sets}\label{sec:specifications}

We next consider disjunctions. A formula of the form
$
\formulaOR := \bigvee_{j=1}^{J} \formulaT_{j}
$
represents \(J\) alternative conjunctions, each associated with a set \(\tvset{}_{\formulaT_j}(t)\). Accordingly,
\begin{equation}
    \tvset{}_\formulaOR(t) = \bigcup_{j=1}^{J} \tvset{}_{\formulaT_j}(t).
\end{equation}
Thus, satisfaction of \(\formulaOR\) can be ensured by rendering forward invariant any one of the sets \(\tvset{}_{\formulaT_j}(t)\). In summary, the feasible region associated with the considered STL fragment is the union of the convex regions induced by the elementary formulas and their conjunctions.

\section{Graph Representation of STL Specifications and Collision Free Navigation}\label{sec:graph_repr} 

From Section~\ref{sec:stl2tvsets}, each STL formula \(\psi\) is associated with a collection of convex sets in the joint time--configuration space, including funnel sets, robustified predicate sets, and, when tasks overlap, intersections thereof. A trajectory satisfies \(\psi\) if it traverses these sets according to the logical and temporal structure of the STL formula. This naturally motivates a graph-based abstraction of the planning problem.

To this end, we introduce two graph structures. The first is a \emph{task graph}, which encodes the convex sets induced by the STL specification \(\psi\) as described in Section~\ref{sec:stl2tvsets}, together with the admissible transitions among them dictated by the temporal organization of the tasks (see Fig.~\ref{fig:gsc_stl_example}). The second is a \emph{collision-free graph}, associated with a convex decomposition of the free configuration space, which encodes admissible collision-free transitions between spatial regions. The product of these two graphs then yields a graph whose vertices represent convex spatio-temporal collision-free regions and whose edges encode admissible logical and geometric progressions. This product graph provides the discrete structure underlying the trajectory-planning formulation developed in the next section.


\subsection{Task graph representation}

To encode the logical structure and temporal progression of the STL specification, we organize the convex regions into a directed graph.

\begin{definition}[Task graph]
Given an STL specification \(\psi\), we define a directed \emph{task graph}
\(
\graph_{\text{task}} = (\vertices_{\text{task}}, \edges_{\text{task}}),
\)
where each vertex \(v \in \vertices_{\text{task}}\) is associated with a convex spatio-temporal set
\(
\sttaskset_v \subseteq \Rpluseq \times \mathcal{Q},
\)
equal to either a funnel set, a predicate set, or a nonempty intersection of such sets arising from overlapping formulas in \(\psi\).
\end{definition}


In particular, let us associate with each task \(\formulaNT_\taskID\) the possible labels \(\mathrm{F}\), \(\mathrm{P}\), and \(\mathrm{C}\), corresponding, respectively, to the funnel, predicate, and completed stages. To each label \(\sigma_\taskID\in\{\mathrm{F},\mathrm{P},\mathrm{C}\}\), we associate a set \(\sttaskset_{\taskID}^{\sigma_\taskID}\), where the labels \(\mathrm{F}\) and \(\mathrm{P}\) correspond to the sets \eqref{eq:funnelset} and \eqref{eq:predicateset}, respectively. The completed label is associated with \(\sttaskset_{\taskID}^{\mathrm{C}} = [t_{\taskID}^{\mathrm{comp}},+\infty)\times \nR{\configdim}\), with \(t_{\taskID}^{\mathrm{comp}}=a_\taskID\) for \(\formulaNT_\taskID=\eventuallyab{a_\taskID}{b_\taskID}\mu^h\), and \(t_{\taskID}^{\mathrm{comp}}=b_\taskID\) for \(\formulaNT_\taskID=\alwaysab{a_\taskID}{b_\taskID}\mu^h\). Thus, \(\sttaskset_{\taskID}^{\mathrm{C}}\) represents the absence of any further constraint associated with task \(\taskID\) after its completion. However, when \(\formulaNT_\taskID\) does not overlap in time with any other task, no explicit completion vertex is required in the task graph. The task is represented by the funnel and predicate vertices, associated with the sets \(\funnelset{\taskID}\) and \(\predicateset{\taskID}\), respectively, while its completion is encoded directly by the outgoing transitions from the predicate vertex.

For each group \(\mathcal{K}_{\mathrm{ov}}^i\), let \(\vertices_{\mathrm{task}}^i \subseteq \vertices_{\mathrm{task}}\) denote the subset of task-graph vertices associated with \(\mathcal{K}_{\mathrm{ov}}^i\). In particular, this also includes the case in which \(\mathcal{K}_{\mathrm{ov}}^i\) is a singleton. Each vertex \(v \in \vertices_{\mathrm{task}}^i\) is associated with a tuple of local labels
$
\sigma(v)=\bigl(\sigma_\taskID(v)\bigr)_{\taskID\in\mathcal{K}_{\mathrm{ov}}^i},
$
and with the corresponding set
$
\sttaskset_v
=
\bigcap_{\taskID\in\mathcal{K}_{\mathrm{ov}}^i}
\mathcal{B}_{\taskID}^{\sigma_\taskID(v)}.
$
Vertices are introduced only for nonempty intersections.

A directed edge \((v,w)\in\edges_{\mathrm{task}}\), with \(v\neq w\), is introduced if the corresponding transition is consistent with the syntax of the STL formula. More precisely, \((v,w)\in\edges_{\mathrm{task}}\) if at least one of the following holds:
\begin{enumerate}
    \item there exists a group \(\mathcal{K}_{\mathrm{ov}}^i\) such that \(v,w \in \vertices_{\mathrm{task}}^i\), \(\sttaskset_v \cap \sttaskset_w \neq \emptyset\), and for every \(\taskID\in\mathcal{K}_{\mathrm{ov}}^i\), the local label evolves by at most one logical step, namely
    \(
    \sigma_\taskID(w)\in\Delta_\taskID\bigl(\sigma_\taskID(v)\bigr),
    \)
    where \(\Delta_\taskID(\mathrm{F})=\{\mathrm{F},\mathrm{P}\}\), \(\Delta_\taskID(\mathrm{P})=\{\mathrm{P},\mathrm{C}\}\), and \(\Delta_\taskID(\mathrm{C})=\{\mathrm{C}\}\);
    \item there exist two consecutive groups \(\mathcal{K}_{\mathrm{ov}}^i\) and \(\mathcal{K}_{\mathrm{ov}}^{i+1}\) such that \(v \in \vertices_{\mathrm{task}}^i\), \(w \in \vertices_{\mathrm{task}}^{i+1}\), \(\sttaskset_v \cap \sttaskset_w \neq \emptyset\), and \(v\) corresponds to completion of all formulas in \(\mathcal{K}_{\mathrm{ov}}^i\).
\end{enumerate}

The task graph also includes a designated source vertex \(\sourceV\), associated with the singleton set \(\{(\tinit,\config_0)\}\), and a target vertex \(\targetV\), which is a dummy vertex associated with completion of the formula. Both are treated as singleton sets in the spatio-temporal space.



\subsection{Collision-free graph}

To encode obstacle avoidance, we assume that the free configuration space \(\freespace\), defined in \eqref{eq:freespace}, is decomposed into finitely many convex regions \(\{\colfreeReg{\ell}\}_{\freeregID=1}^{\freeregnum}\). Such decompositions can be obtained, for instance, via region-inflation or polyhedral-approximation methods \cite{dai2024certified,wang2025fast,werner2025superfast}.
These regions are organized into an undirected graph
\[
\graph_{\text{free}} = (\vertices_{\text{free}}, \edges_{\text{free}}),
\]
where each vertex \(w \in \vertices_{\text{free}}\) is associated with a convex collision-free region \(\colfreeReg{w} \subseteq \mathcal{Q}\), and an edge \((w,w') \in \edges_{\text{free}}\) indicates that the corresponding regions are adjacent and can be traversed without collision. Since the task graph is defined in the joint time--configuration space, each collision-free region is lifted to the spatio-temporal domain as
$$
\spcolfreeReg{w} := \Rpluseq \times \colfreeReg{w},
$$
thereby treating time as an unconstrained coordinate.


\subsection{Product graph under obstacle constraints}

By combining the task graph and the collision-free graph we obtain the product graph
\[
\graph_{\text{prod}} = (\vertices_{\text{prod}}, \edges_{\text{prod}}).
\]
where \(\vertices_{\text{prod}} \subseteq \vertices_{\text{task}} \times \vertices_{\text{free}}\). As defined below, the vertices of \(\graph_{\text{prod}}\) encode convex spatio-temporal regions that are both admissible with respect to the STL specification and contained in the collision-free space, while the edges encode admissible transitions between such regions.

\paragraph*{\underline{Product vertices}}
Let \(u \in \vertices_{\text{task}}\) be a task-graph vertex with associated convex set \(\sttaskset_u \subseteq \Rpluseq \times \configSet\), and let \(w \in \vertices_{\text{free}}\) be a collision-free-graph vertex with associated lifted collision-free region \(\spcolfreeReg{w} \subseteq \Rpluseq \times \configSet\). We then define the product-graph vertex \(v=(u,w)\) with associated set
\[
\gcsset_v \define \sttaskset_u \cap \spcolfreeReg{w}.
\]
The vertex \(v\) is included in \(\vertices_{\text{prod}}\) if and only if \(\gcsset_v \neq \emptyset\). Since both \(\sttaskset_u\) and \(\spcolfreeReg{w}\) are convex, the set \(\gcsset_v\) is convex by construction.
%

\paragraph*{\underline{Product edges}}
Let \(v=(u,w)\) and \(v'=(u',w')\) be two product vertices. We add the directed edge \((v,v')\in\edges_{\text{prod}}\) if the following conditions are mutually satisfied:
\begin{align}
    &u'=u \ \text{or}\ (u,u')\in\edges_{\text{task}},
    \label{eq:prod_edge_task}
    \\
    &w'=w \ \text{or}\ (w,w')\in\edges_{\text{free}},
    \label{eq:prod_edge_free}
    \\
    &\neg(u'=u \ \wedge\ w'=w),
    \label{eq:prod_edge_nobothself}
    \\
    &\gcsset_v \cap \gcsset_{v'} \neq \emptyset .
    \label{eq:prod_edge_setintersection}
\end{align}
Conditions~\eqref{eq:prod_edge_task}--\eqref{eq:prod_edge_nobothself} state that a transition is allowed only if it is admissible in the task graph, in the collision-free graph, or in both, excluding the trivial case in which neither component changes. Condition~\eqref{eq:prod_edge_setintersection} requires the two product vertices to be compatible in the joint time--configuration space. This ensures that continuity constraints between consecutive trajectory segments can be imposed meaningfully in the subsequent optimization.
A path in \(\graph_{\text{prod}}\) therefore identifies a collision-free sequence of admissible spatio-temporal regions satisfying the STL specification.


\begin{remark}
The complexity of the product-graph construction scales with the number of task vertices and the number of collision-free regions. In the worst case, the number of vertices grows as
$
\mathcal{O}\!\left(|\vertices_{\text{task}}|\,|\vertices_{\text{free}}|\right),
$
corresponding to the full Cartesian product of the two graphs, and the number of edges grows accordingly. In practice, this bound is conservative: many task-phase sets do not intersect many collision-free regions, and many candidate transitions are eliminated by empty intersections 
either at vertex level or at the edge level. The resulting product graph is therefore typically much sparser than the worst-case bound suggests.
\end{remark}


%% file: Sections/stl_gcs.tex
\section{Trajectory planning via graphs of convex sets}\label{sec:traj_plan}

The graph constructions introduced in Sect.~\ref{sec:graph_repr} provide a discrete representation of the admissible logical and geometric evolutions of the system. We now endow this graph structure with continuous decision variables and convex constraints, thereby obtaining a GCS formulation of the planning problem. This allows us to cast STL-constrained trajectory planning with obstacle avoidance as a shortest-path problem over convex spatio-temporal regions. We proceed in three steps. First, we formulate the trajectory-planning problem in its infinite-dimensional form, over the space of admissible trajectories. Next, we introduce a B-spline parameterization to obtain a finite-dimensional approximation. Finally, we cast the resulting problem as a shortest-path problem over a graph of convex sets.

\subsection{Trajectory planning problem}

We seek a trajectory satisfying the STL specification
\(
\formulaOR = \bigvee_{\branchidx=1}^{\branchnum} \formulaT_{\branchidx}
\)
while avoiding collisions with obstacles and optimizing a performance criterion. Let \(\mathcal{D}_q \subseteq \mathbb{R}^{\configdim}\) denote the admissible set of configuration velocities, assumed for simplicity to be a box constraint $\mathcal{D}_q:= \{ \|\dot{\config}\|_{\infty} \leq \dot{q}_{max} \}$. Following~\cite{marcucci2023motion}, we introduce a path coordinate \(s\in[0,S]\), where \(S>0\) is a fixed normalization constant (e.g., \(S=1\)), and a strictly increasing time law \(t=p(s)\), satisfying \(p(0)=\tinit\) and \(p(S)=\tfinal\). Defining the spatial trajectory
\(
\spacelaw(s)\define \config(p(s)),
\)
that is, the configuration trajectory expressed in the path coordinate \(s\). The planning problem is then formulated as
\begin{subequations}\label{eq:optimization2}
\begin{align}
    \min_{p(\cdot),\,\spacelaw(\cdot)}
    \quad
    & w_1 \tfinal + w_2 \totallenght(\spacelaw,\tfinal) 
    \label{eq:optimization2_a}
    \\
    \text{s.t.}\quad
    & \spacelaw \circ p^{-1} \in \mathcal{C}^{\flatoutdegree},
    \label{eq:optimization2_b}
    \\
    & \spacelaw(0) = \config_0,
    \label{eq:optimization2_c}
    \\
    & \dot{\spacelaw}(0) = \dot{p}(0)\dot{\config}_0,
    \label{eq:optimization2_d}
    \\
    & \dot{\spacelaw}(s) \in \dot{p}(s)\mathcal{D}_q,
    \qquad \forall s \in [0,S],
    \label{eq:optimization2_e}
    \\
    & \dot{p}(s) > 0,
    \qquad \forall s \in [0,S],
    \label{eq:optimization2_ep}
    \\
    & \tfinal = p(S),
    \label{eq:optimization2_f}
    \\
    & \begin{bmatrix}
        p(s)\\
        \spacelaw(s)
    \end{bmatrix}
    \in
    \bigcup_{v \in \vertices_{\text{prod}}} \gcsset_v
    \qquad \forall s \in [0,S],
    \label{eq:optimization2_g}
\end{align}
\end{subequations}
Here \(\totallenght(\spacelaw,\tfinal)\define \int_{\tinit}^{\tfinal}\norm{\dot{\config}(\tau)}\,d\tau\), and \(w_1,w_2>0\) weight the two objective terms: final time and trajectory length. 
Constraint~\eqref{eq:optimization2_b} imposes the required smoothness of the resulting trajectory in physical time, while \eqref{eq:optimization2_c}--\eqref{eq:optimization2_d} enforce the initial position and velocity conditions. Constraint~\eqref{eq:optimization2_e} encodes the velocity bound in the path-parameterized coordinates and follows from the chain rule,
\(\dot{\config}(t)={\dot{\spacelaw}(s)}/{\dot{p}(s)}\), whereas \eqref{eq:optimization2_ep} guarantees that the time law is strictly increasing. Constraint~\eqref{eq:optimization2_f} defines the final time, and \eqref{eq:optimization2_g} requires the trajectory to evolve in the union of the robustified convex spatio-temporal regions associated with the product graph and, together with continuity of the curve, ensures satisfaction of the STL formula. 
Formulation~\eqref{eq:optimization2} is directly posed in the joint time--configuration space and is therefore compatible with the product-graph construction.


\subsection{B-spline parameterization}

Problem~\eqref{eq:optimization2} is infinite-dimensional, since the decision variables are the functions \(p(\cdot)\) and \(\spacelaw(\cdot)\). We therefore parameterize both using B-splines \cite[Sec.~4.5]{biagiotti2008trajectory}, and choose the spline degree \(\flatoutdegree\) consistently with the continuity required by the planning model. Given control points \(\gamma_{\ctrlpointidx}\), a B-spline curve is written as
\begin{equation}
    \gamma(s)=\sum_{\ctrlpointidx=0}^{\numctrlpoints}
    N_{\ctrlpointidx,\flatoutdegree}(s)\,\gamma_{\ctrlpointidx},
    \label{eq:bspline_def}
\end{equation}
where \(N_{\ctrlpointidx,\flatoutdegree}\) are the corresponding B-spline basis functions. B-splines are particularly convenient here because: 
i) the curve lies in the convex hull of neighboring control points, allowing set-containment constraints to be enforced through the control points; ii) derivatives are again B-splines of lower degree, with derivative control points depending affinely on the original ones; iii) 
initial and final conditions can be enforced directly through the endpoint control points; iv) multiple spline segments can be concatenated with prescribed continuity enforced through linear constraints on the control points.

\subsection{Trajectory optimization over the GCS}

We now formulate problem~\eqref{eq:optimization2} as an instance of the GCS problem~\eqref{eq:gcs_spp} using the B-spline parameterization above. A feasible solution consists of a path on the graph from \(\sourceV\) to \(\targetV\) and a continuous spatio-temporal trajectory obtained by concatenating one B-spline segment per visited vertex, each contained in the corresponding set \(\gcsset_v\).

\paragraph*{\underline{Vertex variables and constraints}}
To each visited vertex \(v\), we associate the spatio-temporal spline
\begin{equation}
    \spatiotemp_{v}(s)
    =
    \begin{bmatrix}
        p_{v}(s)\\
        \spacelaw_{v}(s)
    \end{bmatrix}
    =
    \sum_{\ctrlpointidx=0}^{\numctrlpoints}
    N_{\ctrlpointidx,\flatoutdegree}(s)\,\spatiotemp_{{\ctrlpointidx},v},
\end{equation}
where
$
\spatiotemp_{{\ctrlpointidx},v}
=
(
p_{{\ctrlpointidx},v},
\spacelaw_{{\ctrlpointidx},v})
 \in \Rpluseq \times \configSet$ are the associated control points. 
 We collect them in the stack 
\( 
 \boldsymbol{\zeta}_v \define \operatorname{col} \left( \spatiotemp_{0,v}, \ldots, \spatiotemp_{\numctrlpoints,v} \right). 
\)
Containment of the full segment in \(\gcsset_v
\) is enforced through
\begin{equation}
    \spatiotemp_{{\ctrlpointidx},v} \in \gcsset_v
    ,
    \qquad \forall \ctrlpointidx \in \{0,\dots,\numctrlpoints\},
    \label{eq:vertex_set_membership}
\end{equation}
which is sufficient by the convex-hull property of B-splines.

Velocity constraints are imposed through the derivative control points:
\begin{equation}
    \dot{p}_{{\ctrlpointidx},v} > \dot{p}_{\min} > 0,
    \qquad \forall \ctrlpointidx \in \{0,\dots,\numctrlpoints\},
    \label{eq:time_monotonicity}
\end{equation}
\begin{equation}
    \dot{\spacelaw}_{{\ctrlpointidx},v}
    \in
    \dot{p}_{{\ctrlpointidx},v}\mathcal{D}_q,
    \qquad \forall \ctrlpointidx \in \{0,\dots,\numctrlpoints\},
    \label{eq:velocity_constraint_ctrlpts}
\end{equation}
which are convex because derivative control points depend affinely on the original ones (see e.g.~\cite{csomay2022multi}). Collectively, \eqref{eq:vertex_set_membership}--\eqref{eq:velocity_constraint_ctrlpts} define the vertex constraint set \(\bar{\Xi}_v \in (\Rpluseq \times \configSet)^{J+1}\) for \(\boldsymbol{\zeta}_v\).

\paragraph*{\underline{Edge constraints}}
Each edge \(e=(u,v)\) represents a transition between two consecutive spline segments. Continuity is imposed directly on the control points of the spline and of its derivatives. Let
\(
    \spatiotemp^{[l]}_{\ctrlpointidx,v},
\), $\ctrlpointidx=0,\dots,\numctrlpoints-l$,
denote the control points of the \(l\)-th derivative spline associated with vertex \(v\), with
\(
    \spatiotemp^{[0]}_{\ctrlpointidx,v}
    \define
    \spatiotemp_{\ctrlpointidx,v}.
\)
For a (clamped \cite{biagiotti2008trajectory}) B-spline, equality of the endpoint values of the \(l\)-th derivatives is enforced by equating the last derivative control point of the segment associated with \(u\) with the first derivative control point of the segment associated with \(v\). Thus, for each edge \(e=(u,v)\), we impose
\begin{equation}
    \spatiotemp^{[l]}_{\numctrlpoints-l,u}
    =
    \spatiotemp^{[l]}_{0,v},
    \qquad
    l=0,\dots,\flatoutdegree.
    \label{eq:edge_continuity}
\end{equation}
The case \(l=0\) enforces continuity of the spatio-temporal trajectory, while \(l\geq 1\) enforces continuity of its derivatives up to order \(\flatoutdegree\).
Since derivative control points depend affinely on the original control points, \eqref{eq:edge_continuity} defines linear equality constraints on the spline control points.

\paragraph*{\underline{Optimization problem}}

The final time is
\(
    \tfinal
    =
    p_{\numctrlpoints,\targetV},
\)
and the length of each spline segment is upper-bounded, see e.g.~\cite[Sec.~5.4]{marcucci2023motion}, by
\begin{equation}
    \ell_v(\boldsymbol{\zeta}_v)
    =
    \sum_{\ctrlpointidx=0}^{\numctrlpoints-1}
    \left\|
        \spacelaw_{\ctrlpointidx+1,v}
        -
        \spacelaw_{\ctrlpointidx,v}
    \right\|.
    \label{eq:gcs_segment_length_bound}
\end{equation}
The resulting trajectory-planning problem is
\begin{subequations}
\label{eq:spp_instantiated}
\begin{align}
    \min_{\substack{
        \path\in\pathsSet\\
        \boldsymbol{\zeta}_v,\, v\in\vertices_{\mathrm{prod}}
    }}
    \quad
    &
    w_1 \tfinal
    +
    w_2 \sum_{v\in\path} \ell_v(\boldsymbol{\zeta}_v)
    \label{eq:spp_instantiated_objective}
    \\
    \text{s.t.}\quad
    &
    \boldsymbol{\zeta}_v
    \in
    \bar{\Xi}_v,
    \qquad \qquad
    \forall v\in\path,
    \label{eq:spp_instantiated_vertex_constraints}
    \\
    &
    (\boldsymbol{\zeta}_u,\boldsymbol{\zeta}_v)
    \in
    \bar{\Xi}_e,
    \qquad
    \forall e=(u,v)\in\edges_{\path}.
    \label{eq:spp_instantiated_edge_constraints}
\end{align}
\end{subequations}
This is a mixed-integer convex optimization problem. Although NP-hard, it admits a tight convex relaxation based on perspective functions~\cite{marcucci2024shortest}. Following~\cite{marcucci2023motion}, rather than solving \eqref{eq:spp_instantiated} exactly via branch-and-bound, we solve the convex relaxation and then apply a computationally inexpensive rounding procedure, repeated multiple times to improve the quality of the resulting approximate solution.

%% file: Sections/stl_control.tex
\section{Control: STL satisfaction by forward invariance}\label{sec:control}

The planning layer computes a nominal reference trajectory \(\planState_d(t)\) in the planning coordinates \eqref{eq:planning_coordinates}, together with the corresponding sequence of STL-induced time-varying sets \(\tvset{}_v(t)\) for all \(v \in \path\), corresponding to one of the disjunctive branches $\formulaT_\branchidx$. This trajectory is then mapped back to the original state space, yielding the reference \(\state_d(t)\) for the nonlinear system \eqref{eq:nlsys}. In practice, exact execution of the nominal plan cannot be expected because of actuator limits, model mismatch, and disturbances. For this reason, and because the objective is to preserve safety and STL satisfaction rather than merely reduce tracking error, we complement the planning layer with a low-level feedback controller acting on the original nonlinear dynamics \eqref{eq:nlsys}. The controller is designed to render forward invariant, under the closed-loop dynamics, the intersection of the currently active STL-induced set \(\tvset{}_v(t)\) with the collision-free space \(\freespace\) defined in \eqref{eq:freespace}. Within this architecture, trajectory tracking is promoted through the nominal input, while set invariance is enforced (with minimum violation) through CBFs.

\subsection{Control Barrier Functions}

We introduce control barrier functions for the STL-induced time-varying sets and for obstacle-avoidance constraints. In the proposed framework, the relevant constraint functions are generally \emph{time-varying}, \emph{switching}, and \emph{spatially nonsmooth}. They are time-varying because the STL-induced sets depend explicitly on time, switching because the active task set \(\mathcal{I}^{\formulaT}(t)\) defined in \eqref{eq:active_set} may change at discrete instants, and spatially nonsmooth because Boolean set operations are represented through pointwise min and max operators. In particular, intersections of simultaneously active STL-induced sets are encoded through pointwise minimum operators, whereas obstacle avoidance at the control level is encoded through pointwise maximum operators.

\subsubsection{Time-varying nonsmooth control barrier functions}

We now consider barrier functions that are time-varying and nonsmooth with respect to the state. Let \(g_i:\mathbb{R}^{\xdim}\times\mathbb{R}_{\geq 0}\to\mathbb{R}\), \(i\in\rangen{n}\), be continuously differentiable functions, and define
\begin{equation}
    g(\state,t)
    =
    \operatorname{op}_{i\in\rangen{n}} g_i(\state,t),
    \qquad
    \operatorname{op}\in\{\min,\max\}.
    \label{eq:tv_boolean_barrier_function}
\end{equation}
and
\begin{equation}
    \mathcal{C}(t)
    =
    \{\state\in\stateSet \mid g(\state,t)\geq 0\},
    \label{eq:tv_nonsmooth_cbf_set}
\end{equation}
The case \(\operatorname{op}=\min\) encodes intersections of sets, as for simultaneous STL-induced constraints. The case \(\operatorname{op}=\max\) encodes unions of sets, and is used for obstacle avoidance when the safe set is described as the complement of an obstacle represented by multiple smooth inequalities.
Since the nonsmoothness considered here is spatial, the Clarke generalized gradient \cite{cortes2008discontinuous} is taken only with respect to \(\state\), while time is treated as a parameter. We denote this generalized gradient by
\(
    \partial_x g(\state,t).
\)
%

Define the active-index set
\(
    \mathcal{A}(\state,t)
    \define
    \{i\in\rangen{n}\mid g_i(\state,t)=g(\state,t)\}.
\)
For \(g\) defined as in \eqref{eq:tv_boolean_barrier_function}, the Clarke generalized gradient with respect to \(\state\) is obtained from the active smooth components. In particular,
\(
    \partial_x g(\state,t)
    \subseteq
    \operatorname{co}
    \{\nabla_x g_i(\state,t)\mid i\in\mathcal{A}(\state,t)\}.
\)
Therefore, a conservative and implementation-friendly nonsmooth CBF condition is obtained by enforcing the time-varying CBF inequality on all active smooth components.

\begin{definition}[Time-varying Boolean Nonsmooth Control Barrier Function]\label{def:cbfbooleantv}
The function \(g\) in \eqref{eq:tv_boolean_barrier_function} is a \emph{time-varying Boolean Nonsmooth Control Barrier Function} for system \eqref{eq:nlsys} on an open set \(\mathcal{D}\subseteq\stateSet\) containing \(\mathcal{C}(t)\) if there exists a locally Lipschitz extended class-\(\classK\) function \(\alpha\) such that, for every \((\state,t)\in\mathcal{D}\times\mathbb{R}_{\geq 0}\) with \(\state\in\mathcal{C}(t)\), there exists \(\sysinput\in\inputSet\), which, $\forall i\in\mathcal{A}(\state,t)$, satisfies
\begin{equation}
\begin{aligned}
    \frac{\partial g_i (\state,t)}{\partial \state}
    \bigl(\ff(\state)+\gG(\state)\sysinput\bigr)
    +
    \frac{\partial g_i (\state,t)}{\partial t}
    \geq
    -\alpha\bigl(g(\state,t)\bigr).
    \label{eq:tv_boolean_nonsmooth_cbf_condition}
\end{aligned}
\end{equation}
\end{definition}

When \(g\) is continuously differentiable, the active-index set is a singleton and \eqref{eq:tv_boolean_nonsmooth_cbf_condition} reduces to the standard smooth time-varying CBF condition
\begin{equation}
    \frac{\partial g (\state,t)}{\partial \state}
    \bigl(\ff(\state)+\gG(\state)\sysinput\bigr)
    +
    \frac{\partial g (\state,t)}{\partial t}
    \geq
    -\alpha\bigl(g(\state,t)\bigr).
    \label{eq:smooth_tv_cbf_recovered}
\end{equation}
If a measurable and locally bounded feedback law satisfies \eqref{eq:tv_boolean_nonsmooth_cbf_condition}, then \(\mathcal{C}\) is forward invariant~\cite{glotfelter2020nonsmooth,mestres2025safe}.


\subsubsection{Switching barrier functions}

Let
\(
0 \leq t_1 < \cdots < t_N
\)
denote the switching instants. On each interval \([t_s,t_{s+1})\), \(s\in\{0,\ldots,N\}\), the active task set is fixed, and the corresponding function \(g\) is treated as a time-varying nonsmooth barrier function of the form introduced in \eqref{eq:tv_boolean_barrier_function}. At a switching instant \(t_s\), define the left and right limits
\begin{equation}
g_s^-(\state)
\define
\lim_{t\to t_s^-} g(\state,t),
\qquad
g_s^+(\state)
\define
\lim_{t\to t_s^+} g(\state,t).
\end{equation}
To ensure that membership in the STL-induced set is preserved across switches, we impose the following condition
\begin{equation}
g_s^-(\state)\geq 0
\;\Longrightarrow\;
g_s^+(\state)\geq 0,
\qquad
\forall \state\in\stateSet,\;\forall s\in\rangen{N}.
\label{eq:expansivitycbf}
\end{equation}
Equivalently, denoting by
\[
\begin{aligned}
    \mathcal{C}(t_s^-)
=
\{\state\in\stateSet\mid g_s^-(\state)\geq 0\},
\\
\mathcal{C}(t_s^+)
=
\{\state\in\stateSet\mid g_s^+(\state)\geq 0\},
\end{aligned}
\]
condition \eqref{eq:expansivitycbf} can be written as the set-inclusion condition
\begin{equation}
\mathcal{C}(t_s^-)
\subseteq
\mathcal{C}(t_s^+),
\qquad
\forall i\in\rangen{N}.
\label{eq:expansivity_set}
\end{equation}
Thus, any state that belongs to the set immediately before a switch also belongs to the set immediately after the switch. In the STL-induced construction, this is related to the set-expansion condition imposed at changes of the active task set in Def.~\ref{def:setexpansion}.

\begin{proposition}[Forward invariance under switching]
\label{prop:forward_invariance}
Let \(g\) be piecewise defined over the intervals \([t_s,t_{s+1})\), \(s\in\{0,\ldots,N\}\). Suppose that:
\begin{enumerate}
    \item on each interval \([t_s,t_{s+1})\), there exists a feedback control law satisfying the CBF condition \eqref{eq:tv_boolean_nonsmooth_cbf_condition};
    \item at each switching instant \(t_s\), the set-inclusion condition \eqref{eq:expansivity_set} holds.
\end{enumerate}
Then \(\mathcal{C}(t)\) is forward invariant.
\end{proposition}

\begin{proof}
Forward invariance on each interval in between two switches follows from the standard time-varying CBF argument ~\cite{glotfelter2020nonsmooth,mestres2025safe}. At each switching instant, \eqref{eq:expansivity_set} guarantees that any state in \(\mathcal{C}(t_s^-)\) also belongs to \(\mathcal{C}(t_s^+)\), providing a valid initial condition for the next interval. Repeating this argument over successive intervals proves the result.
\end{proof}



\subsection{Low-level control architecture}\label{sec:lowlevelcontrol}

We adopt a \emph{nominal tracking plus safety filter} architecture \cite{ames2016control}. A nominal input \(\sysinput_{\mathrm{nom}}\) is first generated by a trajectory-tracking controller. This input is then modified online by a quadratic program that enforces obstacle avoidance and forward invariance of the active STL-induced set while remaining as close as possible to the nominal command.

Let the active STL set at time $t$, expressed in the original state coordinates $\state$, be
\begin{equation}
    \taskset_{x,\formulaT_\branchidx}(t)
    \define
    \bigl\{\,
        \state \in \stateSet
        \;\big|\;
        \stlcbf{}_{x,\taskID}(\state, t) \ge 0,
        \;
        \forall \taskID\in \mathcal{I}^{\formulaT_\branchidx}(t)
    \,\bigr\},
    \label{eq:active_stl_set}
\end{equation}
where $\stlcbf{}_{x,\taskID}(\state, t) \define \stlcbf{}_\taskID(\config(\state),t)$
and $\formulaT_\branchidx$ is the plan selected by the solution
of~\eqref{eq:spp_instantiated}. The obstacle regions in state coordinates are
\begin{equation}
    \obsReg_{x,\obsID}
    \define
    \bigl\{\state \in \stateSet \mid \obsfun{x,\obsID}(\state) < 0\bigr\},
    \qquad
    \obsfun{x,\obsID}(\state) \define \obsfun{\obsID}(\config(\state)).
    \label{eq:obs_state}
\end{equation}

The functions \(\stlcbf{}_{\state,\taskID}\) and \(\obsfun{\state,\obsID}\) depend on the full state only through the configuration map \(\config(\state)\). Hence, their relative degree with respect to \eqref{eq:nlsys} is inherited from the planning model, namely from the number of integrators between the input and the configuration. We therefore work directly with relative-degree-one barrier functions, obtained either from the original constraints when their relative degree is one, or through High-Order CBFs \cite{tan2021high,xiao2021high} or backstepping CBFs \cite{taylor2022safe,cohen2024safety} otherwise.

Let \(g^b(\state,t)\) denote the Boolean STL-induced barrier function enforced by the safety filter at time \(t\), written as
\(
    g^b(\state,t)
    =
    \min_{k\in\rangen{N_b}} g^b_k(\state,t),
\)
where \(g^b_k(\state,t)\) are the relative-degree-one STL-induced barrier functions associated with the active STL constraints. The corresponding active-index set is
\(
    \mathcal{A}_b(\state,t)
    \define
    \{\,k\in\rangen{N_b}\mid g^b_k(\state,t)=g^b(\state,t)\,\}.
\)
Similarly, for each obstacle \(l\in\rangen{N_o}\) included in the safety filter, let
\(
    g^o_l(\state)
    =
    \max_{i\in\rangen{N_{o,l}}} g^o_{l,i}(\state)
\)
denote the Boolean nonsmooth barrier function encoding avoidance of the \(l\)-th obstacle, where \(g^o_{l,i}(\state)\) are smooth relative-degree-one obstacle-avoidance barrier functions. The corresponding active-index set is
\(
    \mathcal{A}_{o,l}(\state)
    \define
    \{\,i\in\rangen{N_{o,l}}\mid g^o_{l,i}(\state)=g^o_l(\state)\,\}.
\)
In practice, only obstacles within a suitable neighborhood of the robot are included in the safety filter. The STL-induced barriers are piecewise defined in time and inherit the switching instants of the active STL-induced set, which are handled through Proposition~\ref{prop:forward_invariance}. The obstacle-avoidance barriers are time-invariant in the considered setting, but may be Boolean nonsmooth because each obstacle-free set is represented through a pointwise maximum of smooth functions.

\subsubsection{Safety filter}

At each time instant, the control input is obtained as the solution of the quadratic program
\begin{equation}
\begin{aligned}
    \min_{\sysinput,\,
    \boldsymbol{s}}
    \quad &
    \|\sysinput - \sysinput_{\mathrm{nom}}\|_2^2
    + w_s \|\boldsymbol{s}\|_2^2
    \\
    \text{s.t.}
    &
    L_{\ff}g^b_k + L_{\gG}g^b_k\sysinput + \partial_t g^b_k
    \geq -\alpha_k\bigl(g^b_k\bigr) - s_k,
    \\
    &\qquad \forall k \in \mathcal{A}_b(\state,t),
    \\
    &
    \nabla g^o_{l,i}{}^\top\bigl(\ff(\state)+\gG(\state)\sysinput\bigr)
    \geq
    -\alpha_l\bigl(g^o_l(\state)\bigr),
    \\
    &\qquad
    \forall i \in \mathcal{A}_{o,l}(\state),\;\; l \in \rangen{n_o},
    \\
    &
    \sysinput \in \inputSet,
    \qquad
    \boldsymbol{s} \geq \boldsymbol{0},
\end{aligned}
\label{eq:cbf_clf_qp}
\end{equation}
where 
\(\alpha_k\) and \(\alpha_l\) are extended class-\(\mathcal{K}\) functions, 
\(\boldsymbol{s}\in\mathbb{R}^{N_b}\) is a vector of slack variables for the STL barrier constraints used to provide a minimum violation of the STL condition with weight $w_s\gg0$. 

For relative-degree-one constraints, the configuration-level set-expansion condition~\eqref{eq:expansion_condition} directly implies the switching compatibility condition required by Proposition~\ref{prop:forward_invariance}. For higher-relative-degree constraints, this implication is no longer automatic. To see this, let \(b(\config,t)\) denote a configuration-level STL-induced barrier and consider, for a relative degree \(\reldeg\), the HOCBF sequence
\begin{equation}
\begin{aligned}
    &\psi_0(\planState,t)
    \define
    b(\config,t),
    \\
    &\psi_j(\planState,t)
    \define
    \dot{\psi}_{j-1}(\planState,t)
    +
    \alpha_j\bigl(\psi_{j-1}(\planState,t)\bigr),
    \quad
    j \in \rangen{\reldeg-1},
    \label{eq:hocbf_lifted_sequence_switching}
\end{aligned}
\end{equation}
where the functions \(\alpha_j\) are design parameters chosen as sufficiently smooth extended class-\(\classK\) functions. The corresponding lifted set in planning-state space is
\begin{equation}
    \mathcal{C}_{\mathrm{lift}}(t)
    \define
    \bigl\{
        \planState
        \mid
        \psi_j(\planState,t)\geq 0,
        \;
        j\in\{0,\ldots,\reldeg-1\}
    \bigr\}.
    \label{eq:hocbf_lifted_set_switching}
\end{equation}
While the set-expansion condition guarantees that \(\psi_0(\planState,t_s^+)\geq\psi_0(\planState,t_s^-)\) at a switching time \(t_s\), it does not generally imply the same property for the higher-order functions \(\psi_j\), \(j\geq 1\). Indeed, the first lifted barrier depends on derivative terms such as
\(
    \frac{\partial b(\config,t)}{\partial \config} \dot{\config}
    \)
    and
    \(
    \frac{\partial b(\config,t)}{\partial t},
\)
which may change discontinuously at a switch because of a change of active task, active predicate component, or funnel slope. Therefore, configuration-level expansion does not, in general, imply expansion of the HOCBF-lifted set. The same observation applies to backstepping CBF constructions.

Consequently, when HOCBFs or backstepping CBFs are used, the switching compatibility condition should be imposed or verified directly on the lifted set, namely
\begin{equation}
    \mathcal{C}_{\mathrm{lift}}(t_s^-)
    \subseteq
    \mathcal{C}_{\mathrm{lift}}(t_s^+),
    \qquad
    \forall t_s\in \mathcal{T}_s.
    \label{eq:lifted_switching_compatibility}
\end{equation}
This condition may be enforced conservatively through the choice of the funnel parameters and of the HOCBF/backstepping design. For polytopic predicates, a constructive approach similar to \cite[Sec.~V.C]{marchesini2025sampling} can be used to define the parameters of the barrier such that \eqref{eq:lifted_switching_compatibility} holds.

If \eqref{eq:lifted_switching_compatibility} is not guaranteed, the HOCBF constraints may still act as recovery conditions. In particular, if after a switch some lifted barrier value becomes negative, an inequality of the form 
\( \dot{\psi}_j \geq -\alpha_j(\psi_j) \) 
drives \(\psi_j\) back toward the nonnegative set whenever the corresponding QP constraint is feasible. In this case, however, forward invariance across the switching instant is not guaranteed, and the closed-loop STL satisfaction guarantee is temporarily lost. Thus, exact closed-loop preservation of the STL-induced set requires the lifted compatibility condition~\eqref{eq:lifted_switching_compatibility}, whereas recovery-based enforcement should be interpreted as a practical fallback mechanism.


%% file: Sections/exp_results.tex
\section{Experimental and simulation results}\label{sec:exp_res}

We validate the proposed framework through experiments on a planar free-flyer platform and simulations of a quadrotor in a three-dimensional environment. The objective is twofold: first, to show that the proposed planner generates collision-free trajectories satisfying nontrivial STL specifications with moderate computational effort; second, to show that the forward-invariance-based control layer preserves STL satisfaction during execution despite tracking errors induced by model mismatch and disturbances.

In both studies, we report the STL-induced satisfaction margin $\stlcbf{}_{\formulaT}(\config,t)$ associated with the selected branch \(\formulaT\).  Its nonnegativity certifies membership in the active STL-induced time-varying set and, by construction, satisfaction of the selected branch.

\subsection{Planar free-flyer experiment}\label{sec:exp_atmos}

The experimental platform is a planar free-flyer robot moving quasi-frictionlessly on a resin floor, emulating micro-gravity conditions; see Fig.~\ref{fig:atmos}. Its configuration is
\(
\config=(\pos,\yaw)\in\mathbb{R}^2\times(-\pi,\pi],
\)
where \(\pos\in\mathbb{R}^2\) is the planar position and \(\yaw\in(-\pi,\pi]\) is the yaw angle. The control input is the body-frame wrench
\[
\sysinput=
\begin{bmatrix}
\force^\top & \torquez
\end{bmatrix}^\top
\in\mathbb{R}^3,
\]
generated by eight thrusters through the allocation map
\(
\sysinput=\allocmat\thrustersforces
\),
with \(\allocmat\in\mathbb{R}^{3\times 8}\) full row rank. The full nonlinear dynamics are
\(
\ddconfig=\redG(\yaw)\sysinput
\),
where
\begin{equation}
    \redG(\yaw)=
    \begin{bmatrix}
        \frac{1}{m}\rotz(\yaw) & \zeros{2\times 1}
        \\
        \zeros{1\times 2} & \frac{1}{\inertiaz}
    \end{bmatrix},
\end{equation}
\(m\) is the mass, \(\inertiaz\) is the yaw moment of inertia, and \(\rotz(\yaw)\) is the planar rotation matrix. Defining
\(
\state=
\begin{bmatrix}
\config^\top & \dot{\config}^\top
\end{bmatrix}^\top
\),
the dynamics fit the control-affine form \eqref{eq:nlsys} and are feedback linearizable.

The robot is required to satisfy the STL specification
\begin{equation}
\begin{aligned}
    \formulaOR
    &=
    \formulaT_1 \vee \formulaT_2 \vee \formulaT_3 \vee \formulaT_4,
    \\
    \formulaT_1
    &=
    \eventuallyab{12}{26}\pred{h_1}
    \wedge
    \eventuallyab{34}{52}\pred{h_4}
    \wedge
    \alwaysab{62}{70}\pred{h_5}
    \wedge
    \eventuallyab{84}{112}\pred{h_3},
    \\
    \formulaT_2
    &=
    \eventuallyab{12}{26}\pred{h_2}
    \wedge
    \eventuallyab{32}{50}\pred{h_4}
    \wedge
    \eventuallyab{60}{82}\pred{h_5}
    \wedge
    \alwaysab{94}{102}\pred{h_3},
    \\
    \formulaT_3
    &=
    \eventuallyab{14}{30}\pred{h_4}
    \wedge
    \eventuallyab{38}{58}\pred{h_1}
    \wedge
    \alwaysab{68}{75}\pred{h_5}
    \wedge
    \eventuallyab{90}{116}\pred{h_3},
    \\
    \formulaT_4
    &=
    \eventuallyab{14}{30}\pred{h_2}
    \wedge
    \alwaysab{38}{47}\pred{h_1}
    \wedge
    \eventuallyab{58}{80}\pred{h_5}
    \wedge
    \eventuallyab{88}{116}\pred{h_3}.
\end{aligned}
\label{eq:stl_atmos}
\end{equation}
The predicates \(h_i\) are associated with the convex goal regions in the planar workspace shown in Fig.~\ref{fig:map_partioning}. Static obstacles are accounted for both at the planning level, through the collision-free decomposition, and at the control level, through the CBF-based safety filter.
The planner selects branch \(\formulaT_3\) of \eqref{eq:stl_atmos} and computes a collision-free spatio-temporal trajectory satisfying the corresponding sequence of tasks. The computational times of the main planning stages are reported in Table~\ref{tab:atmos_planning_times}. The total planning time is \(33.9426\,\mathrm{s}\), with the largest contribution arising from the construction of the product graph.

The planned trajectory is tracked online using the nominal tracking controller and safety-filter architecture introduced in Section~\ref{sec:lowlevelcontrol}. Figure~\ref{fig:map2d_and_traj} shows the collision-free decomposition and the trajectory executed by the robot. Despite tracking errors caused by model mismatch and actuation imperfections, the robot avoids the obstacles and completes the selected STL branch.
Figure~\ref{fig:atmos_rob} shows the STL-induced satisfaction margin
\(\stlcbf{}_{\formulaT_3}(\config,t)\). Its nonnegativity throughout the experiment confirms that the executed trajectory remains inside the active STL-induced time-varying set. The discontinuities in the plot correspond to switching instants at which a task is completed and the active-task set changes.

\begin{table}[t]
    \centering
    \caption{Planning times for the planar free-flyer experiment.}
    \label{tab:atmos_planning_times}
    \begin{tabular}{lc}
        \hline
        \textbf{Stage} & \textbf{Time [s]} \\
        \hline
        STL/task-graph construction & \(0.6244\) \\
        Product-graph construction & \(18.6375\) \\
        GCS optimization & \(14.6807\) \\
        \hline
        Total & \(33.9426\) \\
        \hline
    \end{tabular}
\end{table}

\begin{figure}
    \centering
    \includegraphics[width=0.4\linewidth]{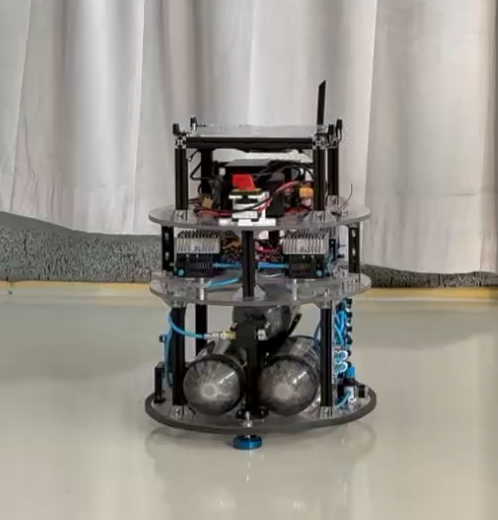}
    \caption{Planar free-flyer robot used in the experiments.}
    \label{fig:atmos}
\end{figure}

\begin{figure}[t]
    \centering
    \begin{subfigure}[t]{0.48\columnwidth}
        \centering
        \includegraphics[width=\textwidth]{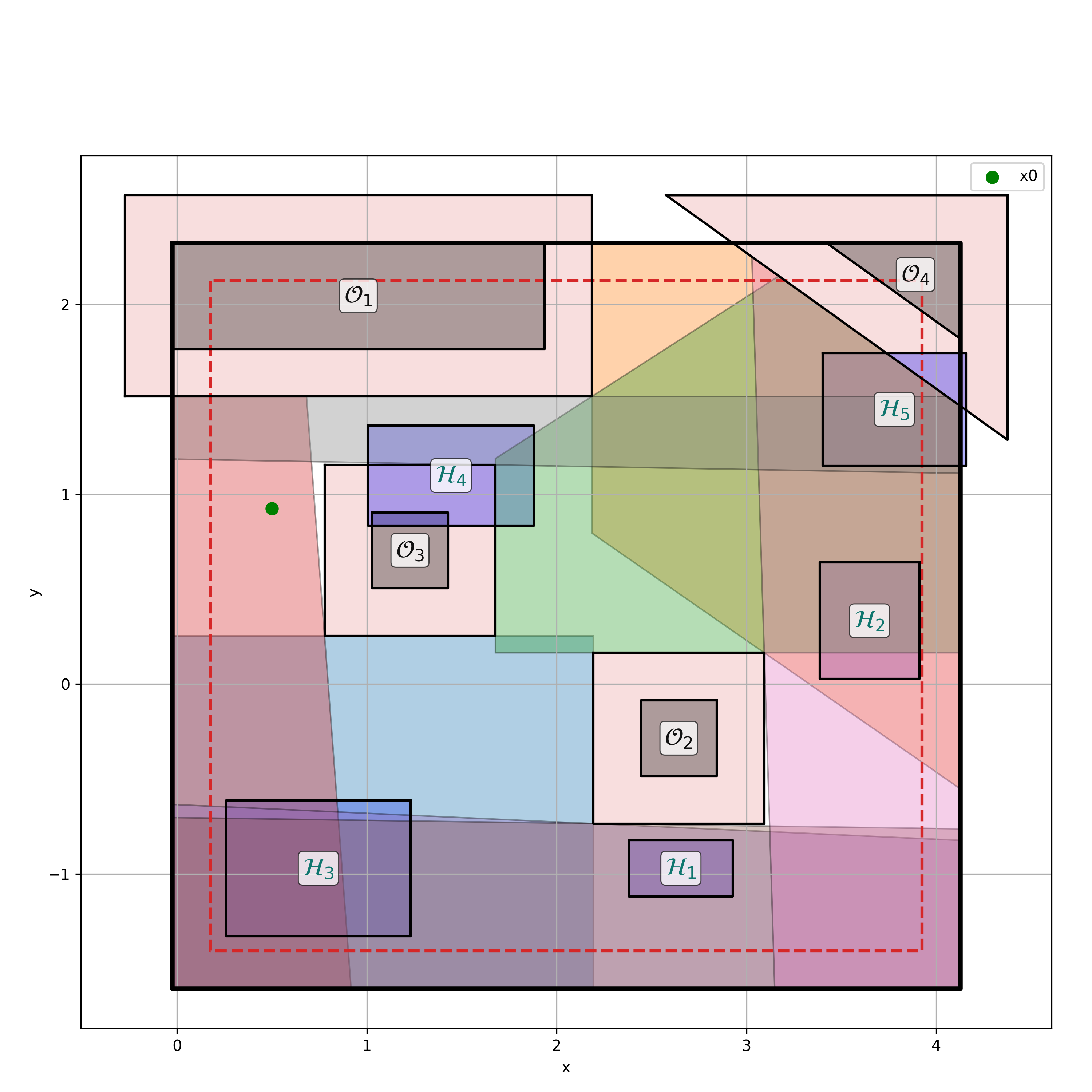}
        \caption{Workspace and free-space partition.}
        \label{fig:map_partioning}
    \end{subfigure}
    \hfill
    \begin{subfigure}[t]{0.48\columnwidth}
        \centering
        \includegraphics[width=\textwidth]{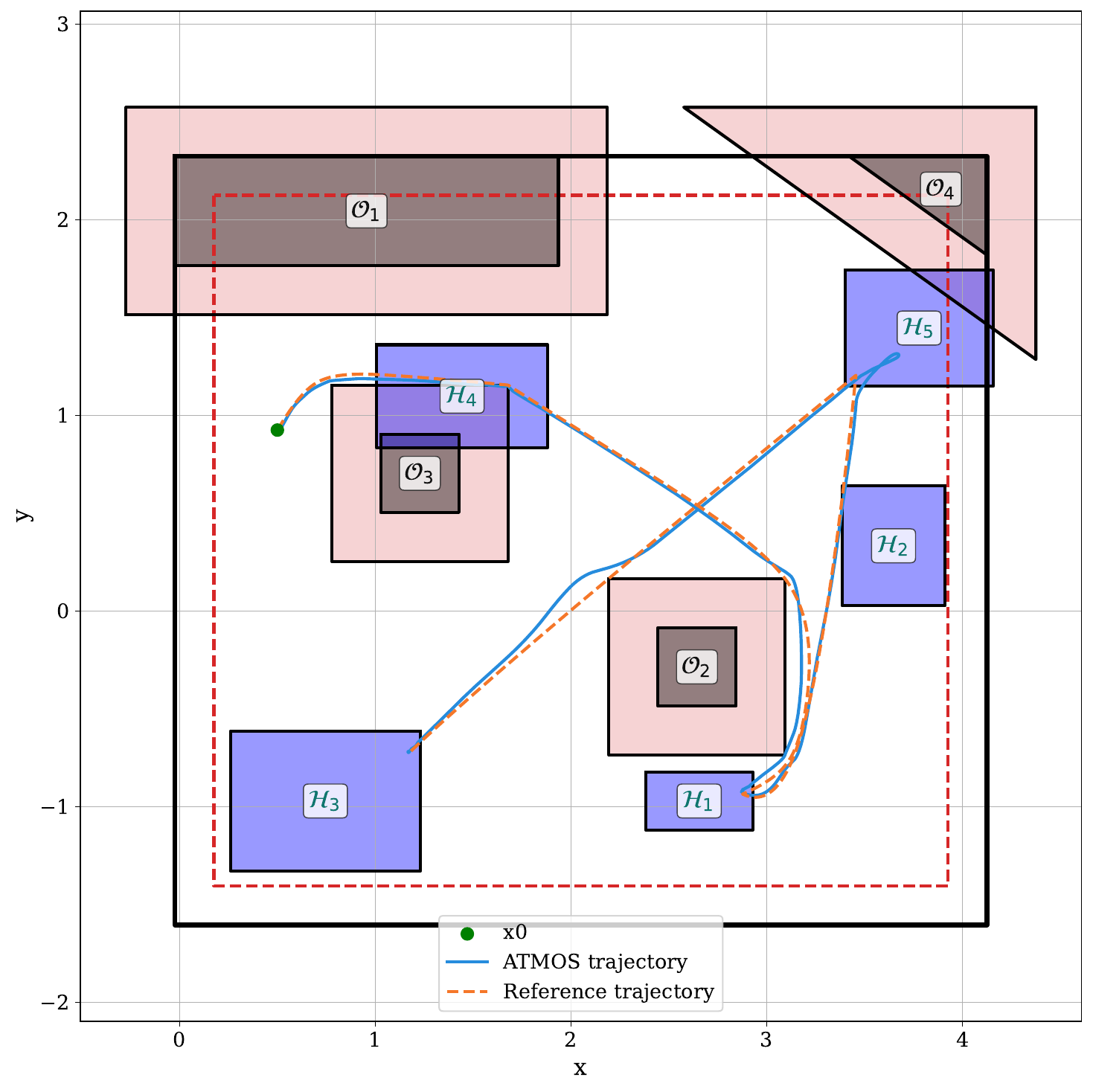}
        \caption{Executed trajectory.}
        \label{fig:atmos_traj}
    \end{subfigure}
    \caption{Planar free-flyer experiment. Left: laboratory workspace, with obstacles in gray, obstacle inflation in pink, workspace-boundary inflation in dashed red, and the collision-free convex partition obtained using IRIS~\cite{dai2024certified}. Right: trajectory executed by the robot.}
    \label{fig:map2d_and_traj}
\end{figure}

\begin{figure}
    \centering
    \includegraphics[width=0.9\linewidth]{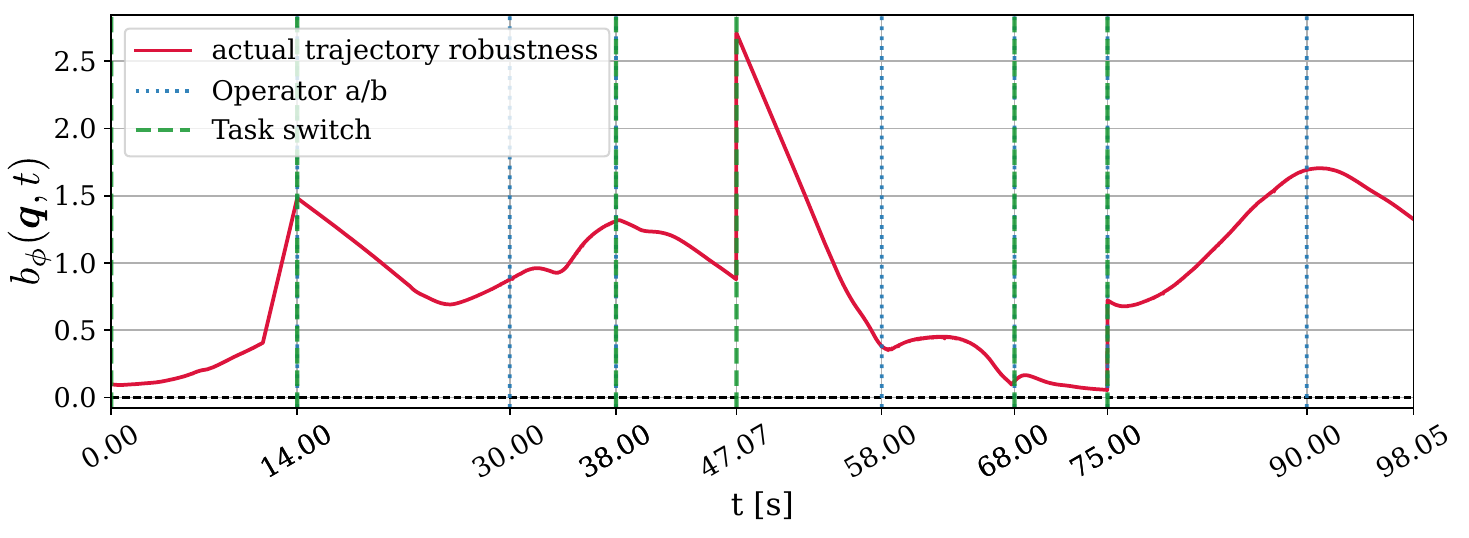}
    \caption{Temporal evolution of the STL-induced satisfaction margin
\(\stlcbf{}_{\formulaT_3}(\config,t)\) during the planar free-flyer experiment.}
    \label{fig:atmos_rob}
\end{figure}

\subsection{3D quadrotor simulation study}\label{sec:3dsim}

We next evaluate the proposed framework in a simulated three-dimensional quadrotor scenario. This study considers a higher-dimensional configuration space while retaining the same planning and control architecture used in the planar experiment.

The environment consists of the workspace
\(
    \mathcal{X}
    =
    [-10,10]\times[-10,10]\times[0,10],
\)
containing four box-shaped obstacles and six convex goal regions; see Fig.~\ref{fig:quad_3dtraj}. The obstacles are inflated by \(0.2\,\mathrm{m}\) to account for the dimensions of the quadrotor, and the resulting free space is partitioned into collision-free polytopic regions using IRIS~\cite{dai2024certified}.

We consider the STL specification
\begin{equation}
\begin{aligned}
    \formulaOR
    &=
    \formulaT_1 \vee \formulaT_2 \vee \formulaT_3 \vee \formulaT_4,
    \\
    \formulaT_1
    &=
    \eventuallyab{8}{22}\pred{h_1}
    \wedge
    \eventuallyab{24}{42}\pred{h_6}
    \wedge
    \eventuallyab{48}{68}\pred{h_3}
    \wedge
    \eventuallyab{75}{98}\pred{h_1}
    \\
    &\wedge
    \alwaysab{108}{128}\pred{h_2},
    \\
    \formulaT_2
    &=
    \eventuallyab{8}{24}\pred{h_6}
    \wedge
    \eventuallyab{34}{58}\pred{h_1}
    \wedge
    \alwaysab{68}{90}\pred{h_2}
    \wedge
    \eventuallyab{100}{122}\pred{h_3},
    \\
    \formulaT_3
    &=
    \eventuallyab{8}{22}\pred{h_1}
    \wedge
    \eventuallyab{32}{58}\pred{h_4}
    \wedge
    \alwaysab{70}{92}\pred{h_3}
    \wedge
    \eventuallyab{104}{130}\pred{h_2},
    \\
    \formulaT_4
    &=
    \eventuallyab{8}{24}\pred{h_6}
    \wedge
    \eventuallyab{30}{50}\pred{h_5}
    \wedge
    \alwaysab{60}{82}\pred{h_3}
    \wedge
    \eventuallyab{92}{116}\pred{h_1}
    \\
    &\wedge
    \eventuallyab{126}{150}\pred{h_2}.
\end{aligned}
\label{eq:quad_3d_stl}
\end{equation}

The predicates \(h_i\) are associated with the convex goal regions shown in Fig.~\ref{fig:quad_3dtraj}. As in the planar experiment, obstacle avoidance is incorporated at the planning level through the collision-free decomposition and at the control level through the CBF-based safety filter.

The planner uses the robustness margin \(\robustMargin=0.1\). The GCS problem is solved with MOSEK using the convex relaxation and randomized rounding over at most \(12\) candidate paths. The planner selects branch \(\formulaT_2\) of \eqref{eq:quad_3d_stl}.

The computational times of the main planning stages are reported in Table~\ref{tab:quad_planning_times}. The total planning time is \(58.4177\,\mathrm{s}\), with the largest contribution arising from the GCS optimization. Compared with the planar experiment, the increased computational cost reflects the higher-dimensional spatio-temporal search space.

The planned trajectory is tracked using the full nonlinear quadrotor dynamics together with the nominal tracking controller and safety-filter architecture introduced in Section~\ref{sec:lowlevelcontrol}. To evaluate the closed-loop behavior under uncertainty, the simulated plant includes parametric mismatch and additive input disturbances applied before actuator saturation. Specifically, for the commanded input
\(
    \sysinput
    =
    [T,\tau_x,\tau_y,\tau_z]^\top,
\)
the disturbed input is modeled as
\[
    \sysinput_d
    =
    \sysinput+w_u,
    \qquad
    w_u
    \sim
    \mathcal{N}\!\left(
        0,
        \operatorname{diag}\!\left(
            \sigma_T^2,
            \sigma_{\tau_x}^2,
            \sigma_{\tau_y}^2,
            \sigma_{\tau_z}^2
        \right)
    \right),
\]
where
\[
    \sigma_T=0.10\,\mathrm{N},
    \qquad
    (\sigma_{\tau_x},\sigma_{\tau_y},\sigma_{\tau_z})
    =
    (0.01,\,0.01,\,0.008)\,\mathrm{N\,m}.
\]

Figure~\ref{fig:quad_3dtraj} shows the executed trajectory. Despite the imperfect tracking due to model mismatch and input disturbances, the quadrotor avoids the obstacles, and completes the selected branch \(\formulaT_2\).

Figure~\ref{fig:quad_rob} shows the corresponding STL-induced satisfaction margin
\(\stlcbf{}_{\formulaT_2}(\config,t)\). As in the planar experiment, its nonnegativity throughout the simulation confirms that the executed trajectory remains inside the active STL-induced time-varying set and therefore satisfies the selected branch.

\begin{table}[t]
    \centering
    \caption{Planning times for the 3D quadrotor simulation.}
    \label{tab:quad_planning_times}
    \begin{tabular}{lc}
        \hline
        \textbf{Stage} & \textbf{Time [s]} \\
        \hline
        STL/task-graph construction & \(0.4364\) \\
        Product-graph construction & \(13.3572\) \\
        GCS optimization & \(44.6241\) \\
        \hline
        Total & \(58,4177\) \\
        \hline
    \end{tabular}
\end{table}

\begin{figure}
    \centering
    \includegraphics[width=0.80\linewidth]{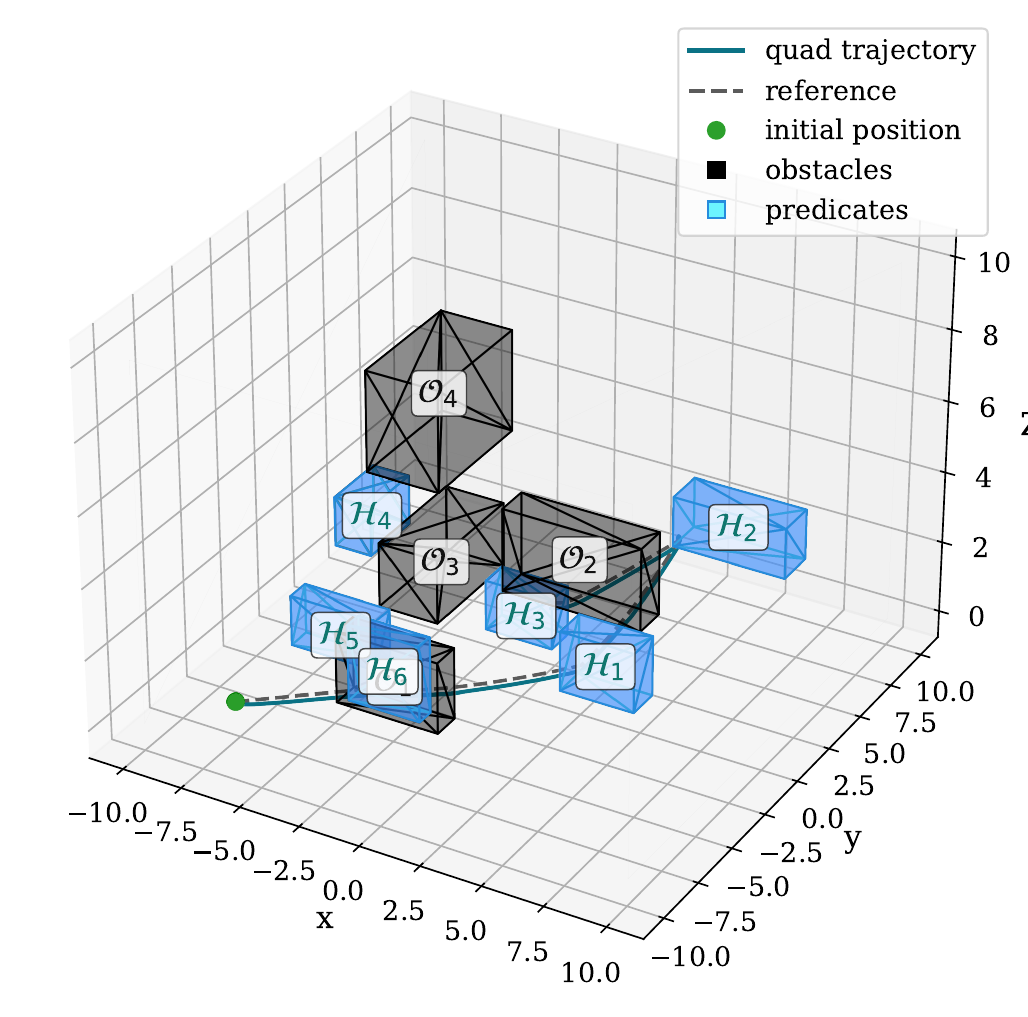}
    \caption{Executed 3D quadrotor trajectory satisfying the selected STL branch \(\formulaT_2\). The trajectory visits the prescribed goal regions while avoiding the obstacles.}
    \label{fig:quad_3dtraj}
\end{figure}

\begin{figure}
    \centering
    \includegraphics[width=1.0\linewidth]{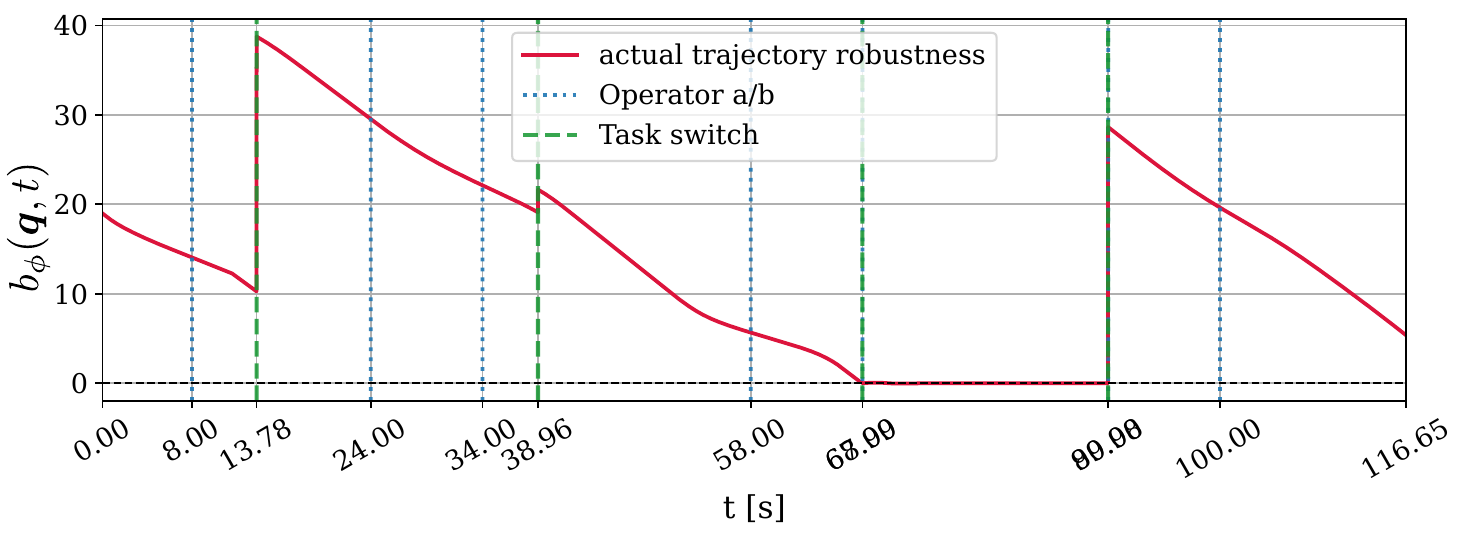}
    \caption{Temporal evolution of
    $\stlcbf{}_{\formulaT_2}(\config, t)$ during the 3D quadrotor simulation. The nonnegative robustness confirms satisfaction of the selected STL branch \(\formulaT_2\) along the executed trajectory.}
    \label{fig:quad_rob}
\end{figure}

%% file: Sections/conclusions.tex
\section{Conclusions}\label{sec:conclusion}

We presented a unified planning and control framework for satisfying a fragment of Signal Temporal Logic specifications in obstacle-rich environments. The proposed method encodes STL tasks through time-varying and spatio-temporal sets, organizes their convex execution phases into a graph structure, and combines this representation with a collision-free decomposition through a product graph. This yields a Graph of Convex Sets formulation of the planning problem, in which B-spline trajectory segments are optimized over convex spatio-temporal regions.

At the control level, the same STL-induced admissible sets are enforced through a CBF-based feedback architecture, providing a direct connection between high-level task planning and low-level forward-invariance-based execution.

The approach was validated in real-world experiments on a free-flyer robot and in a 3D quadrotor simulation study. The results demonstrate the ability of the framework to generate collision-free trajectories satisfying nontrivial STL specifications and to preserve task satisfaction during execution.

%% file: Sections/app_selectingkappa.tex
\subsection{Computing $\bar{\lintermtv}_{\taskID}$ by containment LPs}\label{app:computingkappa}

Consider two consecutive tasks $\taskID$ and $\taskID+1$, with associated 
set 
$\tvsetspatiotemp_{\taskID}$ and funnel set $\funnelset{\taskID+1}$.
We assume both sets are described in $\mathcal{H}$-representation, consistently with the notation
introduced in Section~\ref{subsec:stl_frag}. In particular,
{
\begin{align}
%
\tvset{}_{\taskID}(t) &= 
\left\{
\config \in \configSet
\;\middle|\;
\begin{aligned}
&\bigl[\; \cC_{\taskID} \;\; \tfrac{\bar{\lintermtv}_{\taskID}}{\alpha_{\taskID}}\,\ones{\numconsh{\taskID}} \;\bigr]
\begin{bmatrix} \config \\ t \end{bmatrix} 
\\
&\leq \dd_{\taskID} + (\bar{\lintermtv}_{\taskID} - \robustMargin)\,\ones{\numconsh{\taskID}}, 
\end{aligned} \label{eq:funnelpoly}
\right\}.
\\[6pt]
%
\funnelsetslice{\taskID+1}(t) &= 
\left\{
\config \in \configSet
\;\middle|\;
\begin{aligned}
&\bigl[\; \cC_{\taskID+1} \;\; \tfrac{\bar{\lintermtv}_{\taskID+1}}{\alpha_{\taskID+1}}\,\ones{\numconsh{\taskID+1}} \;\bigr]
\begin{bmatrix} \config \\ t \end{bmatrix} 
\\
&\leq \dd_{\taskID+1} + (\bar{\lintermtv}_{\taskID+1} - \robustMargin)\,\ones{\numconsh{\taskID+1}}, 
\end{aligned} \label{eq:funnelpoly}
\right\}.
\end{align}
}
Our objective is to determine the smallest slope parameter $\bar{\lintermtv}_{\taskID+1} \ge 0$ such that the 
set $\tvset{}_{\taskID}(t)$ 
is contained in the funnel set of task $\taskID+1$ at time 
$t=\tinittask{\taskID+1}$, as required by the
containment condition~\eqref{eq:expansion_condition}.

\paragraph*{\underline{Containment via LPs}}  
Recall that 
$\tvset{}_{\taskID}(\tinittask{\taskID+1}) \subseteq \funnelsetslice{\taskID}(\tinittask{\taskID+1})$
if and only if every inequality defining the outer polytope $\funnelsetslice{\taskID+1}
(\tinittask{\taskID+1})$ is satisfied over the inner polytope 
$\tvset{}_{\taskID}(\tinittask{\taskID+1})$.

Let $\cc_{(\taskID+1),i}$ denote the i-th row of the constraint matrix $\cC_{\taskID+1}$. After fixing 
$t=\tinittask{\taskID+1}$, the $i$-th right-hand side of the funnel inequality becomes:
\begin{equation}
    \bar{d}_{\taskID+1,i}(\bar{\lintermtv}_{\taskID+1}) 
    = d_{\taskID+1,i} + \bar{\lintermtv}_{\taskID+1}\!\left(1 - \tfrac{
    \tinittask{\taskID+1}}{\alpha_{\taskID+1}} \right) - \robustMargin.
\end{equation}
Containment therefore requires that, for each constraint $i$,
\begin{equation}
    \sup_{\config \in 
    \tvset{}_{\taskID}(\tinittask{\taskID+1})} 
    \cc_{(\taskID+1),i}^\top \config
    \;\leq\;
    \bar{d}_{\taskID+1,i}(\bar{\lintermtv}_{\taskID+1}).
    \label{eq:polytope_containment}
\end{equation}
Since 
$\tvset{}_{\taskID}(\tinittask{\taskID+1})$ is a polytope, each supremum in \eqref{eq:polytope_containment} can be efficiently computed by solving a linear program.
Specifically, for each row $i$ of $\cC_{\taskID+1}$ we define
\begin{equation}
\begin{aligned}
    \nu_i \define 
    \max_{\config} \quad & \cc_{(\taskID+1),i}^\top \config \\
    \text{s.t.} \quad & \cC_{\taskID} \config 
    + \tfrac{\bar{\lintermtv}_{\taskID}}{\alpha_{\taskID}}\,\ones{\numconsh{\taskID}}\tinittask{\taskID+1}  \leq \dd_{\taskID} 
    + (\bar{\lintermtv}_{\taskID} - \robustMargin)\,\ones{\numconsh{\taskID}}.
\end{aligned}
\label{eq:lp_support}
\end{equation}

Containment then reduces to the set of scalar inequalities
\begin{equation}
    \nu_i
    \;\leq\; d_{\taskID+1,i} + \bar{\lintermtv}_{\taskID+1}\!\left(1 - \tfrac{
    \tinittask{\taskID+1}}{\alpha_{\taskID+1}} \right) - \robustMargin,
    \quad \forall i.
    \label{eq:containment_constraints}
\end{equation}

\paragraph*{\underline{Final optimization}}  
The minimal feasible $\bar{\lintermtv}_{\taskID+1}$ is obtained by solving the single-variable linear program
\begin{equation}
\begin{aligned}
    \min_{\bar{\lintermtv}_{\taskID+1}} \quad & \bar{\lintermtv}_{\taskID+1} \\
    \text{s.t.} \quad & \nu_i 
    \leq d_{\taskID+1,i} + \bar{\lintermtv}_{\taskID+1}\!\left(1 - \tfrac{
    \tinittask{\taskID+1}}{\alpha_{\taskID+1}} \right) - \robustMargin,
    \quad \forall i,
\end{aligned}
\label{eq:containment_gamma}
\end{equation}

Since the term $\left(1 - \tfrac{
\tinittask{\taskID+1}}{\alpha_{\taskID+1}} \right)$ is constant and strictly positive (as $
\tinittask{\taskID+1}< \alpha_{\taskID+1}$), we can solve explicitly the inequality for $\bar{\lintermtv}_{\taskID+1}$:
$$\bar{\lintermtv}_{\taskID+1} \geq \frac{\nu_i + \robustMargin - d_{\taskID+1,i}}{1 - \tfrac{
\tinittask{\taskID+1}}{\alpha_{\taskID+1}}}, \quad \forall i.$$

The minimal admissible slope is therefore
\begin{equation}
    \bar{\lintermtv}_{\taskID+1} 
    = \max\left( 0, \quad \max_i \frac{\nu_i + \robustMargin - d_{\taskID+1,i}}{1 - \tfrac{
    \tinittask{\taskID+1}}{\alpha_{\taskID+1}}} \right).
\end{equation}
The non-negativity constraint accounts for the case in which $
\tvset{}_{\taskID}(\tinittask{\taskID+1})$ is already strictly contained in $\funnelsetslice{\taskID+1}(
\tinittask{\taskID+1})$, in which case a negative value would correspond to an expansion of the funnel rather than a contraction.

\begin{remark}  
The derivation above relies on \(\mathcal{H}\)-representations of polytopes. For more general convex sets, the containment condition
\(
\tvset{}_{\taskID}(\tinittask{\taskID+1}) \subseteq \funnelsetslice{\taskID+1}(
\tinittask{\taskID+1})\)
can be expressed in terms of support functions:
\[
\sigma_{
\tvset{}_{\taskID}(\tinittask{\taskID+1})}(\boldsymbol{c}) \le \sigma_{\funnelsetslice{\taskID+1}(
\tinittask{\taskID+1})}(\boldsymbol{c}), \quad \forall \boldsymbol{c} \in \mathbb{R}^n,
\]
where \(\sigma_\mathcal{C}(\boldsymbol{c}) = \sup_{x \in \mathcal{C}} \boldsymbol{c}^\top x\) denotes the support function of set $\mathcal{C}$ \cite[Section~3.2.3]{boyd2004convex}.   

For certain special convex sets, this can be written explicitly, e.g.:  
\begin{itemize}
    \item \textbf{Sphere:} If $\mathcal{C} = \{ \config \mid \|\config-\config_0\|_2 \le r \}$, then
    $\sigma_\mathcal{C}(\boldsymbol{c}) = \boldsymbol{c}^\top \config_0 + r \|\boldsymbol{c}\|_2$.
    \item \textbf{Ellipsoid:} If $\mathcal{C} = \{ \config \mid (\config-\config_0)^\top \boldsymbol{Q}^{-1} (\config-\config_0) \le 1 \}$, then
    $\sigma_\mathcal{C}(\boldsymbol{c}) = \boldsymbol{c}^\top \config_0 + \sqrt{\boldsymbol{c}^\top \boldsymbol{Q} \boldsymbol{c}}$.
\end{itemize}

In the general case, one would need to check the containment for all directions \(\boldsymbol{c}\), which is infinite-dimensional and therefore computationally intractable. In practice, a conservative approach is to replace general convex sets with polytopic approximations: an \emph{outer polytope} for \(\robpredset_{\taskID}\) and an \emph{inner polytope} for \(\funnelsetslice{\taskID+1}(
\tinittask{\taskID+1})\), which reduces the problem to the polytopic LP formulation discussed above.  
\end{remark}
